%% file: Arxiv-version.tex
\documentclass[format=acmsmall]{acmart}

\acmConference{}{}{} %
\startPage{0} %
\AtBeginDocument{
\fancypagestyle{firstpagestyle}{%
	\fancyhf{}%
}
\fancypagestyle{acmec}{%
	\fancyhf{}%
        \fancyhead[C]{}
	\fancyhead[R]{\thepage}%
}
\pagestyle{acmec}
\renewcommand{\footnotetextcopyrightpermission}[1]{} %
\renewcommand{\footnotetextauthorsaddresses}[1]{} %
}

\usepackage{booktabs} %

\setcitestyle{authoryear}
\usepackage{times}
\usepackage{soul}
\usepackage{url}
\usepackage[utf8]{inputenc}
\usepackage{graphicx}
\usepackage{tikz}
\usepackage{booktabs}
\usepackage[switch]{lineno}
\usepackage{amsfonts}
\usepackage{amsthm}
\usepackage{amsmath}
\usepackage{times}
\usepackage[ruled]{algorithm2e}
\usepackage{comment}
\LinesNumbered 
\SetAlgoSkip{smallskip}
\usepackage{enumitem}
\usepackage{setspace}
\usepackage{thmtools}
\usepackage{thmtools}
\usepackage{thm-restate}
\usepackage[bblfile=output]{biblatex-readbbl}

\newtheorem{theorem}{Theorem}[section]
\newtheorem{lemma}[theorem]{Lemma}
\newtheorem{definition}[theorem]{Definition}
\newtheorem{claim}{Claim}
\newtheorem{fact}{Fact}
\newtheorem{corollary}[theorem]{Corollary}
\newtheorem{observation}[theorem]{Observation}
\newtheorem{remark}{Remark}
\newtheorem{example}{Example}[section]

\SetAlFnt{\small}
\SetAlCapFnt{\small}
\SetAlCapNameFnt{\small}
\SetAlCapHSkip{0pt}
\IncMargin{-\parindent}
\renewcommand{\vec}[1]{\mathbf{#1}}

\newcommand{\ignore}[1]{}

\title{Pool Formation in Oceanic Games: Shapley Value and Proportional Sharing
}

\begin{document}

\author{Aggelos Kiayias}\affiliation{\institution{University of Edinburgh \& Input Output Global (IOG)} \country{UK}}
\author{Elias Koutsoupias} \affiliation{\institution{University of Oxford} \country{UK}}
\author{Evangelos Markakis} \affiliation{\institution{Athens University of Economics and Business \& Input Output Global (IOG)}\country{Greece}}
\author{Panagiotis Tsamopoulos} \affiliation{\institution{Athens University of Economics and Business}\country{Greece}}

\begin{abstract}
We study a game-theoretic model for pool formation in Proof of Stake blockchain protocols. In such systems, stakeholders can form pools as a means of obtaining regular rewards from participation in ledger maintenance, with the power of each pool being dependent on its collective stake.
 The question we are interested in is the design of mechanisms, ``reward sharing schemes,'' that suitably split rewards among pool members  and achieve favorable properties in the resulting pool configuration. 
 With this in mind, we initiate a non-cooperative game-theoretic analysis of the well known Shapley value scheme from cooperative game theory into the context of blockchains. In particular, we focus on the {\it oceanic} model of games, proposed by Milnor and Shapley (1978), which is suitable for populations where a small set of large players coexists with a big mass of rather small, negligible players. This provides an appropriate level of abstraction for pool formation processes that occur among the stakeholders of a blockchain.
 We provide comparisons between the Shapley mechanism and the more standard proportional scheme, in terms of attained decentralization, via a Price of Stability analysis and in terms of  
   susceptibility to Sybil attacks, i.e., the strategic splitting of a players' stake with the intention of participating in multiple pools for increased profit.
 Interestingly, while the widely deployed proportional scheme appears to have certain advantages, the Shapley value scheme,  which rewards higher the most pivotal players, emerges as a competitive alternative, by being able to bypass some of the downsides of proportional sharing in terms of Sybil attack susceptibility, while also not being far from optimal guarantees w.r.t. decentralization.
 Finally, we also complement our study with some variations of proportional sharing, where the profit is split in proportion to a superadditive or a subadditive function  of the stake, showing that our results for the Shapley value scheme are maintained in comparison to these functions as well. 
\end{abstract}

\maketitle

\newpage

\section{Introduction}
\label{sec:intro}

Permissionless blockchain protocols base participation on the resources that parties possess such as their {\em   computational power} (as in Bitcoin, \cite{Nakamoto2008}), or their {\em stake}  in the system, measured by the number of digital coins they own (as in Ethereum \cite{Ethereum}, see e.g.,  \cite{casper-incentives}). This resource-based operation opens up the possibility for pooling resources together and having multiple resource holders engage as a single entity in protocol operation. Pooling resources can have both positive effects for the system, such as reducing the variance of rewards awarded to participants, as well as negative ones, leading to centralization with a handful of large pools controlling the protocol. 
Importantly, viewing the blockchain protocol as a mechanism, the question that arises is what are the  objectives of this mechanism design problem and how it is possible to realize them. 

As a running paradigm, we will use throughout our work the pool formation process in Proof of Stake blockchains. In such systems, stakeholders are attracted to coalitions as a means of obtaining rewards from participation in block production. The power of such a (stake) pool is then dependent on the collective stake (up to some threshold) that its members possess, which affects the probability that a pool operator is selected as a validator or block producer. Creating or joining a pool can be done either explicitly via options that the protocol itself provides (referred to as {\it onchain pooling}, for example, in the Cardano blockchain) or by agreements among the stakeholders via smart contracts or even via informal agreements. In fact for the Ethereum protocol, there even exist platforms\footnote{To give an idea of their volume and popularity, their total market capitalization is currently over 30 billion USD.} (e.g. Rocketpool or Lido) that facilitate the grouping of stakeholders into pools. Essentially these platforms match smaller stakeholders to larger ones, so as to collectively accrue the amount of 32 ETH, required for a block validator.

Once a pool forms, the central question of interest, which is the focus of our work as well, is how should the pool members split the received rewards among themselves. This has led to an interesting and currently active research agenda, with several (and sometimes conflicting) desiderata that one would like to satisfy. 
Among these, a significant desideratum in blockchain communities is to promote decentralization. Viewed as a game, this means that we would like the reward scheme to induce equilibria with a relatively high number of pools. At the same time, we would also like to achieve resilience to Sybil attacks \cite{douceur02}, since the pseudonymity of blockchain systems may allow seemingly independent pools controlled by the same entity. Although completely eliminating this may be too much to hope for, cf. \cite{kwon2019impossibility}, one could aim for mechanisms that attempt to discourage Sybil behavior.  

If we look at current implementations, both in onchain pooling and in mediating platforms, 
the deployed reward scheme is typically very simple: the pool operator 
may keep a certain amount (accounting for a profit margin and/or operational costs) and the remaining amount is split in a proportional manner, based on stake contribution. This puts forth the natural question of whether the proportional scheme is the best thing to do.
To answer this, it seems reasonable to also consider alternative candidate solutions from economic theory. Concepts from cooperative game theory, like the core \cite{gillies1953some} or the Shapley value \cite{Shapley53}, to name a few, have been extensively studied %
in almost any setting that involves cost or profit sharing.
Despite however the popularity of such notions within game theory and economics, there have been surprisingly very few attempts to incorporate them in models tailored to blockchain protocols, such as \cite{LB+15,CSSZ20} (discussed further in our related work section). More generally, there is only a limited number of works that provide comparisons among different reward schemes, with some notable exceptions being \cite{brunjes2020reward,chen2019axiomatic,CHP22}.

\subsection{Our contribution}
The main question we are interested in is:

\vspace{6pt}
    \noindent\emph{How does the choice of a reward sharing mechanism within a pool affect performance w.r.t. decentralization, and w.r.t. resilience to Sybil attacks?} 
\vspace{6pt}

To study this question, we consider the model of \emph{oceanic games}, originally introduced by \cite{MS78} in the context of cooperative games. This model captures populations where (atomic) players with relatively large stakes coexist with smaller non-atomic players of negligible stake (the ocean). We believe this model is particularly suitable for certain blockchain environments, as also advocated in \cite{LLP19}. Empirical studies have shown that the stake distribution in many blockchain systems follows a power-law distribution. %
In such systems, a small number of participants hold a significant portion of the total stake, while the majority of users control only a small fraction. This aligns well with the oceanic game framework, where large stakeholders influence the system alongside a large number of small participants.

Within this model, we study and compare two reward schemes; the first one is based on the \emph{Shapley value}, as one of the most prominent concepts of profit sharing from game theory. The second one is \emph{proportional sharing}, used extensively in practice.

We analyze the corresponding \emph{non-cooperative games}, induced by these reward schemes, where each player chooses which pool to join.
We quantify the attained decentralization through a \emph{Price of Stability} analysis on the set of Nash equilibria. The Price of Stability measures the quality of the \emph{best} Nash equilibrium by comparing it to the ideal non-selfish solution. We employ the Price of Stability rather than the more pessimistic Price of Anarchy, as it is a more appropriate measure in our setting (where, as we will exhibit, some bad but quite unnatural equilibria may arise). The Price of Stability corresponds to an equilibrium that maximizes social welfare, and players will likely attempt to converge to it (or near it). Conversely, the Price of Anarchy corresponds to an equilibrium that minimizes welfare, and it is improbable that players would select this equilibrium.
With this in mind, we show that the proportional scheme always has an optimal equilibrium (i.e. the Price of Stability equals 1). Not far from this, the Shapley value scheme attains a constant Price of Stability bounded by 4/3. We complement this with an analysis for the purely atomic model, where the Shapley scheme deteriorates, albeit not by a lot, achieving a bound 
of $2$. 

We also consider the resiliency of such schemes against \emph{Sybil attacks}, where the players can split their stake and contribute to multiple pools. Although none of the schemes can completely avoid Sybil attacks, we highlight a potential drawback of the proportional scheme, which can be avoided by the use of the Shapley value. In particular, we prove that under the equilibria of the Shapley scheme identified in the previous sections, the players have no incentive to perform Sybil attacks. On the contrary this is not true for proportional sharing,
which is demonstrated with rather simple examples. %

Finally, we include a study of some non-linear variations of the proportional scheme. Specifically, we consider the \emph{proportional to square roots} and \emph{proportional to squares} mechanisms, which use subadditive and superadditive reward functions, respectively. We exhibit that both have their drawbacks, regarding our desiderata. Namely the proportional to squares scheme has unbounded Price of Stability whereas the proportional to square roots scheme is more vulnerable to Sybil attacks.

Overall, while our results reveal certain advantages of the proportional scheme, 
they also discover important drawbacks (cf. Sybil attacks above). At the same time, the Shapley value
emerges as a competitive alternative, that can bypass some of the negative aspects of proportional splitting,
while maintaining  comparable guarantees w.r.t. the Price of Stability for decentralization. 

\subsection{Related Work}

In terms of the model we use, oceanic games were introduced by \cite{MS78}. In their work, they define the Shapley value and study its properties for a class of cooperative games, namely weighted voting games, when the set of players contains both atomic and non-atomic players. Especially for blockchain systems, the work of \cite{LLP19} revisited and advocated the use of oceanic games for mining. They also studied further aspects concerning the stability of the grand coalition of all players. In our work we consider a non-cooperative model of oceanic games, where we care to evaluate the Nash equilibrium outcomes (i.e., partitions into pools).

Our work is related to the literature of weighted voting games in social choice theory. This is a class of cooperative games, where the value of a coalition is threshold-based, and reward sharing is viewed as determining the voting power of each player. The main solutions that have been proposed for such games are the core, the Shapley value and the Banzhaf index, and we refer to \cite{CW16} for an overview of results. The main difference from our model is that we have a non-cooperative game where each 
strategy profile essentially induces a collection of weighted voting games, one per pool.

The use of solution concepts from cooperative game theory in blockchain applications has been limited so far. In \cite{LB+15} the notion of the core was studied for a model of mining in Bitcoin. They obtained a negative result that the core is empty and hence there is always a motive for some players to deviate. The Shapley value as a reward scheme in mining pool games has also been proposed in \cite{CSSZ20}. The focus of that work was however on its computational aspects and no game-theoretic analysis was given. 

Some closely related works, along a similar spirit are \cite{brunjes2020reward} and \cite{marmolejo}, which present different reward sharing schemes for the Proof of Stake setting. 
Our difference is that they focus on how the system should allocate rewards to the pools, which are then split in a proportional manner among pool members (after the operator takes a cut). In our case we study an orthogonal question of  how the pool should allocate the received rewards among its members. 
Hence it can be seen as  complementary direction to \cite{brunjes2020reward}.

There are also several papers discussing decentralization from various angles. The work of \cite{AH22} introduces a model where the utility function of the players is decentralization conscious, in terms of effort exerted. In \cite{BGR24} a different model is analyzed where effort corresponds to committed stake, and is related to Tullock contests. Another attempt more tailored to Proof of Work systems is presented in \cite{AW22}, which highlights the heterogeneity of the costs invested in hardware. The main difference of all these papers with our work is that they do not study explicitly any pool formation process, as their models focus on the individual effort or investment on resources at equilibrium.  

One of the few works that formally studies the properties of proportional sharing for blockchain protocols is  \cite{chen2019axiomatic}. This is an axiomatic study, tailored for Proof of Work protocols, and demonstrates why the proportional rule can have favorable performance (but in a model without pool formation). Other variants that are studied include the proportional to square roots and the proportional to squares schemes that we also consider. A follow up work with a further axiomatic study is provided in \cite{CHP22}. These models 
however do not consider any pool formation aspects, and there is no discussion on the performance of these schemes w.r.t. decentralization. 

Finally, Sybil attacks have been a major concern for any environment where identity cannot be traced effectively \cite{douceur02}. In the context of auctions, it is often referred to as ``false-name'' bidding 
\cite{AM01}. It is natural that blockchains also offer opportunities for such manipulations, and general resilience against Sybil attacks seems a utopia \cite{kwon2019impossibility}. Nevertheless some positive results have also been obtained towards having reward schemes that can potentially deter such behaviors, e.g. \cite{brunjes2020reward}.

\section{Model and definitions}

\subsection{The oceanic model with atomic and non-atomic players}

We consider a population of stakeholders, say $N$, in a Proof of Stake protocol, who will be referred to from now on as the players of the underlying game. We will mostly focus on the model of oceanic games, as defined by \citet{MS78}, where the population is split into two types of players. Namely, $N = N_a\cup N_s$, 
where $N_a = \{1,\dots, n\}$ is a set of $n$ atomic players. Each $i\in N_a$, possesses stake equal to $a_i$. 
The remaining players are {\it small} players, each with a tiny stake $\epsilon > 0$. We let $\epsilon \rightarrow 0$, and we can view $N_s$ as a continuum of infinitesimally small, non-atomic players (i.e., the ocean). We denote by $L$ the total stake possessed by the ocean, which equals the measure of $N_s$. For convenience we imagine $N_s$ as arranged in the interval $[0, L]$. With this in mind, any subinterval $I\subseteq [0, L]$ corresponds to a mass of players possessing stake equal to the length of $I$. 
This model shares some characteristics with the atomic and non-atomic models in congestion/routing games~\cite{roughgarden16}.

\noindent {\bf Available strategies.} Every player is considering either to operate a pool or join the pool of some other player. The corresponding game is as follows:
every player $i$ needs to choose some index $j\in N$. If she chooses herself, i.e., $j=i$, it means that she is starting her own pool, otherwise she joins someone else's pool.
A strategy profile $\vec{x}$ specifies a choice of strategy $x_i$ for each player $i\in N$. We say that a strategy profile is {\it valid} if whenever some player $i$ chooses $x_i=j$, with $i\neq j$, we also have $x_j = j$, i.e., users who do not operate a pool do choose a valid pool to join. 
Technically, non-valid profiles could also arise in this model, but we will exclude them from our equilibrium analysis, as they lack a natural real life interpretation. %

\noindent {\bf Pool rewards.} A valid strategy profile induces a partition of the players into disjoint sets (pools), $\Pi = (S_1,\dots S_k)$, for some $k$, so that $\cup_j S_j = N$. Given such a partition, each pool $S_j$ receives a reward 
$\rho(S_j)$, determined by the execution of the blockchain protocol. 
In this work, we consider a simple reward function, as an attempt to model the rationale behind the operation of Ethereum or other protocols with a similar design. 
In Ethereum, a stake holder needs to have 32 ETH 
in order to register as a validator and claim rewards. This gives rise to a threshold-based scheme, where a pool can obtain rewards when its total stake exceeds a given threshold $h$.
To define this formally, for a given pool $S\subseteq N$, let $m(S)$ be the total stake that it possesses. This consists of two terms; the total mass of stake by the non-atomic players, which equals $|S\cap N_s|$, and the stake of the atomic players, which is equal to $\sum_{i\in S\cap N_a} a_i$. Then

\begin{align}
\label{eq:rho}
  \rho(S) &=
                 \begin{cases}
                   1, & \text{if }  |S\cap N_s|+\sum_{i\in S\cap N_a} a_i \geq h \\
                   0, & \text{otherwise}
                 \end{cases}
\end{align}

We say that a pool $S\subseteq N$ is a winning pool if $\rho(S) = 1$.

\begin{remark}
\label{rem:stake}
    For the remainder of the paper, we make the assumption that $a_i< h$, for any atomic player $i\in N_a$. The main reason is that, as also considered in the initial model of oceanic games by \cite{MS78}, we want to focus on players who need to collaborate with other people in order to form a successful pool. If there exist players with $a_i> h$, given our threshold-based reward function $\rho$, intuitively it makes sense for them to break their stake and run their own pools with avatars of size $h$ each, and participate in our game with the remainder $a_i ~mod~ h$. For a further discussion on this, see Section \ref{sec-app:big-players} in the Appendix. 
\end{remark}

\noindent {\bf Reward sharing scheme within a pool.} 
Once the pools are formed, the pool members need to split the obtained reward according to some agreed upon scheme. 

\begin{definition}[Reward sharing schemes]
Given a pool $S\subseteq N$ and its reward $\rho(S)$, a reward sharing scheme\footnote{Given the form of $\rho(S)$, each pool can be seen as a weighted voting game \cite{CW16} and the reward to each pool member can be also interpreted as her power within the pool.} specifies a payment allocation $p_i(S)$ to each pool member $i\in S$, dependent on the stake that $i$ contributes to $S$, so that all the payments sum up to  $\rho(S)$. 
\end{definition}

\begin{remark}
We focus on scenarios where the same reward scheme is applied to all formed pools. This is the case for example in protocols with onchain pooling (e.g., Cardano) but also in many of the platforms that facilitate pooling (e.g., for Ethereum). 
We do not consider secondary effects where pool members may re-negotiate their payments later on, among themselves. 
\end{remark}

In our model, we have not explicitly defined the cost that a pool may bear towards participating in the protocol. We essentially assume that $\rho(S)$ is the reward that is left for being shared, after subtracting the amount that the pool operator may claim to cover her operational costs.

\subsection{Equilibria, Decentralization and Price of Stability}
Under any strategy profile, we define the utility of each player to be the reward that she receives.
For a non-valid profile, the reward of a player who does not end up in a pool is simply equal to 0. 
We are interested in profiles that induce a partition into winning pools, and we will denote by $u_i(\Pi)$ the utility of a player under a formed partition $\Pi = (S_1,\dots, S_k)$. 

\begin{definition}
    Given a game $G$, a partition $\Pi$ into winning pools is a Nash equilibrium if for every player $i\in N$, 
    $u_i(\Pi) \geq u_i(\Pi')$, for any partition $\Pi'$ that arises from $\Pi$ if $i$ moves to a different pool or opens a new pool on her own. 
\end{definition}

A major concern in blockchain design is whether the actual execution of the protocol by selfish entities can result in a decentralized configuration. As a metric for this, we use the number of winning pools that are formed at an equilibrium. Given a game, the ideal scenario is that the players split into pools of size exactly $h$. Since this may not always be possible,
we let $OPT(G)$, for a game $G$, be the maximum number of winning pools that can form over all possible 
partitions, i.e., 

\begin{equation}
\label{eq:OPT}
    OPT(G) = 
    \max \{ t: \exists \text{  $\Pi = (S_1,\dots S_t)$ s.t. } m(S_i) \geq h  ~~\forall i\in [t] \}
\end{equation}

In the purely atomic case, when $N_s=\emptyset$, finding $OPT(G)$ is a dual version of the well known Bin Packing problem, referred to as dual Bin Packing in \cite{AJKL84}, and is easily seen to be intractable.

An interesting question that arises is whether equilibria can lead to an approximate solution. In order to evaluate Nash equilibria, the usual metrics are the Price of Anarchy (PoA) and the Price of Stability (PoS). The Price of Anarchy is not an appropriate metric here as it is unavoidable to have completely centralized equilibria, which are however rather unrealistic. This is captured below.

\begin{claim}
\label{mycl:PoA}
    When $a_i< h$ for all $i\in N_a$, the grand coalition consisting of all players is an equilibrium under any reward scheme, and hence the Price of Anarchy is $\Omega(n)$.
\end{claim}

We will therefore turn our attention to the Price of Stability, defined as the ratio between the quality of the optimal solution and the quality of the best Nash equilibrium.
Given a partition $\Pi$, let $W(\Pi)$ be the number of winning pools that are formed under $\Pi$. 

\begin{definition}
    The Price of Stability (PoS) of a game $G$ is computed as $\min_{\Pi\in NE} \frac{OPT(G)}{W(\Pi)}$, where $NE$ is the set of all equilibrium partitions. 
\end{definition}

The Price of Stability appears to be an appropriate quality benchmark for the equilibria of the games we consider. As the Price of Stability aligns with maximizing social welfare, the dynamics of threshold-based mechanisms would likely guide players towards more decentralized outcomes.

Finally, we note that in the remainder of the paper, any missing proofs can be found in the Appendix.

\section{Analysis of the Shapley value reward sharing scheme}
\label{sec:shapley}

In this section, we study the scheme, where the reward in each winning pool is distributed using the Shapley value of each player. 
The Shapley value was introduced by \citet{Shapley53} as the unique reward allocation rule that satisfies a particular set of axioms in cooperative games. Ever since, there is hardly any economic setting involving the sharing of costs or payoffs where the Shapley value has not been considered. Prominent examples are the design of truthful cost-sharing mechanisms \cite{MS01} and the study of power indices in voting games \cite{CW16}. Furthermore, despite the fact that the exact computation of the Shapley value is a \#P-hard problem \cite{DP94}, there are very efficient sampling-based approximations that work quite well \cite{BMRPRS10,fatima2008linear}. Even further, the Shapley value has been deployed successfully in practice, in machine learning applications, as a way to identify important parameters in the training of neural networks \cite{RW+22}.

\subsection{Definition of the Shapley value in the oceanic model}
\label{subsec:Shapley-def}

The standard definition and use of the Shapley value is for the purely atomic model, when $N_s=\emptyset$. It is very instructive to recall first how the Shapley scheme works there, without any small players.

\noindent {\bf The atomic model.} The rationale for the Shapley value is that within each pool, \emph{each stakeholder gets their expected  marginal contribution} if the members of the pool would arrive in a random order. Hence, one needs to take the average over all permutations. The Shapley value for a player $i$, when she belongs to a pool $S$ is therefore given by the following formula:
\begin{align}
\label{eq:shapley}
    \phi_i(S)=\sum\limits_{T\subseteq S\setminus \{i\}}\frac{|T|!(|S|-|T|-1)!}{|S|!}(\rho(T\cup \{i\})-\rho(T))
\end{align}

This is a valid scheme since it is shown already in Shapley's original work that $\sum_{i\in S} \phi_i(S) = \rho(S)$.

\begin{example}
\label{example-PoS}
Consider a game $G$ with only atomic players, with stake distribution $\vec{s}=(3,1,1,1,1,1)$
and let the threshold be $h=4$. 
The optimal pool formation is to have two pools, each
with total weight $4$. Thus $OPT(G)=2$.
However, this is not a Nash equilibrium,
because the large player prefers to switch to the other pool: the reward of the large player in the pool with one small player (with stake vector $(3,1)$) is $1/2$, while if she joins the other pool with the 4 small players, so that the pool becomes $(3,1,1,1,1)$, it is $3/5$. Therefore, only the grand coalition can be a Nash equilibrium (which is a rather extreme case, since we will see that in bigger games other equilibria also arise). \qed
\end{example}

\noindent {\bf Extension to the oceanic model.}
We discuss now how to compute the Shapley value, when we include non-atomic players. Consider a pool with non-atomic players of total mass $k$ and $t$ large players with stakes $(a_1, \dots, a_t)$; the total stake of the pool is $k + a_1 + \dots + a_t$. Imagine that the small players are arranged in the interval $[0, k]$. Then a random arrival of the players corresponds to placing arrival times $L_{i_1}, L_{i_2},\dots, L_{i_t}$ for the atomic players on the interval $[0, k]$. The arrivals can be viewed as filling in the stake $k + a_1 + \dots + a_t$, say starting at $0$, from left to right, as follows: first, a mass of small players (possibly empty) arrives, filling in the interval $[0, L_{i_1}]$; then the first large player arrives, so that the stake of the current subset of players accrues to $L_{i_1} + a_{i_1}$. This is followed by another mass of small players; then a second large player arrives at $L_{i_2}$, accumulating now a total stake of $L_{i_2} + a_{i_1}+ a_{i_2}$; and so on. 
For random arrivals, the points $L_{i_1}, L_{i_2},\dots, L_{i_t}$ are \emph{uniformly distributed within the interval of length $k$}. The reward of each large player $i$ is then the probability that the threshold $h$ falls within the interval that makes player $i$ pivotal. 
To be more precise, for a given arrival order, and for an atomic player $i$, let $P(i)$ be the set of other atomic players preceding
$i$ in the order. Then, if $L_i$ is the random  arrival time of player $i$, the Shapley value is given as follows:
\begin{equation}
\label{eq:shapley-oceanic}
    \phi_i(S) = Pr\left[\sum_{j\in P(i)} a_j + L_i   < h \leq  \sum_{j\in P(i)} a_j + L_i + a_i\right]
\end{equation} 

After computing the reward of all large players, the remaining reward is distributed evenly among the small players. Hence, for the small players, all that matters  is the reward per unit of stake, which is what we will consider in the game-theoretic analysis that will follow. 

\subsection{The Price of Stability in the oceanic model}

For a pool with many large players along with a mass of small players, the exact formula for the Shapley reward might be too complex to derive analytically. For our study on the Price of Stability however, it will suffice to only consider equilibria in which every pool has at most one large player. Hence, when arguing about equilibria and deviations to other pools, it is sufficient to be able to compute the rewards of pools with at most 2 large players.
This is done in the following lemma, which we will use repeatedly. 

\begin{lemma} \label{lemma:nonatomic-rewards}
  Consider a pool with a mass of non-atomic players equal to $k$ and total
  stake at least $h$.  
  \begin{itemize}
      \item If $k\leq h$ and the pool has a single large player with
  stake $a$, the Shapley value of this player is  $(k+ a-h)/k$. 
  \item If $k > h$ and the pool has a single large player with  stake $a$, the Shapley value of this player is  $a/k$. 
  \item If $k\leq h$ and the pool has two
  large players with stakes $a_1,a_2$, respectively, the Shapley value of the player with stake $a_1$
  is equal to (with an analogous formula for the second player)
  \[
   \frac{(h-a_2)^2-(\max(0, h-a_1-a_2))^2+ (\max(0, a_1-h+k))^2}{2k^2}.
 \]
\end{itemize}
 
In all 3 cases above, the remaining reward is distributed equally to the non-atomic
 players.
\end{lemma}
    
\begin{proof}
Consider a pool with a single large player and with $k\leq h$. We apply the procedure described in Section \ref{subsec:Shapley-def} on how to determine the reward of the large players. If $L_1$ is the random arrival of the single player within the interval $[0, k]$, then the player is pivotal, only when $L_1\in [h-a,k]$. The probability of this event
  is $(k-(h-a))/k$, which proves the first part of the lemma.

Suppose now that $k> h$. In that case, the large player is pivotal, only when $L_1\in [h-a, h)$, and this even occurs with probability $a/k$.

Finally, for the last part of the lemma, the formula for a pool of two large players,  is a bit more
  complicated. If $L_1, L_2$ are their random arrival times, we distinguish two cases: $L_1\leq L_2$ and $L_2<L_1$. In
  the first case, player 1 arrives before the second player, and hence she receives a reward when
  $L_1\in [h-a_1,k]$. In the second case, player 2 arrives before player 1, which means that for player 1 to be pivotal, it should hold that $L_1 + a_2\in[h-a_1, h]$. This implies that
  $L_1\in [h-a_1-a_2, h-a_2] \cap [0,k]$. The ranges for the values of $L_1, L_2$, under which player 1 is pivotal, are depicted in
  Figure~\ref{fig:2largeplayers}. The probability that a random point
  $(L_1,L_2)$ falls into these areas gives the reward of player 1.

\begin{figure}
  \centering
  \scriptsize
  \begin{tikzpicture}[scale=4]
    \draw (0,0) rectangle (1,1);

    \draw[->] (0,0) -- (1.1,0) node[right] {$L_1$};
    \draw[->] (0,0) -- (0,1.1) node[above] {$L_2$};

    \foreach \x/\xlabel in {0/0, 0.25/$h-a_1-a_2$, 0.85/$h-a_2$, 0.6/$h-a_1$, 1/$k$} {
        \draw (\x,0) -- (\x,-0.02) node[below] {\xlabel};
        \draw (0,\x) -- (-0.02,\x) node[left] {\xlabel};
    }

    \fill[gray!50] (0.25,0) -- (0.85,0) -- (0.85,0.85) -- (0.25,0.25) -- cycle;
    \fill[gray!50] (0.6,1) -- (1,1) -- (0.6,0.6) -- cycle;

    \draw[thick] (0,0) -- (1,1);

\end{tikzpicture}
\caption{\footnotesize Calculation of the Shapley value for the first player in a winning pool of
  two large players with stakes $a_1$, $a_2$ and non-atomic mass of
  $k\leq h$. It is easy to verify that player 1 is
  pivotal when $(L_1, L_2)$ is in the gray areas. The gray triangle shows the pivotal cases when $L_1 \leq L_2$, whereas the trapezoid corresponds to the cases with $L_1\geq L_2$. The reward of player 1 is then the gray area (equal to
  $((h-a_2)^2-\max(0, h-a_1-a_2)^2+\max(0, a_1-h+k)^2)/2$) divided by the
  total area (equal to $k^2$). When $h-a_1>k$, the gray triangle
  disappears. Similarly, when $h-a_1-a_2<0$, the gray trapezoid becomes
  a triangle.}
  \label{fig:2largeplayers}
\end{figure}
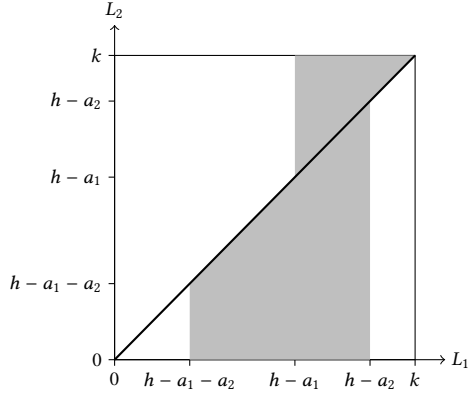

\end{proof}

We exhibit below an example to demonstrate 
how we can construct good equilibria in this game.

\begin{example}
\label{example:non-atomic}
Fix the threshold at $h=3$ and consider a population where all large players have stake $a=2$. Suppose also that there is enough mass of small players, so that we can create a partition into two types of pools. The first type has a single large player and non-atomic players
with total mass $k=2$ (hence each such pool has total stake equal to $4$). The second type has only non-atomic players with mass
$l=4$. Let us argue that this partition is a Nash equilibrium.

The reward of each large player, as given by
Lemma~\ref{lemma:nonatomic-rewards}, is $(a-h+k)/k=1/2$. If the player
switches to a pool with only non-atomic players, her reward (by applying now the second part of Lemma \ref{lemma:nonatomic-rewards}) will be
$a/l=1/2$, so the player has no reason to switch. 
Also, if the player
switches to a pool with another large player, her rewards will be at
most $1/2$, since she will share the reward with another identical
large player, and there is also a mass of non-atomic players, who also receive some part of the profit. Furthermore, we can verify that for the non-atomic players, the reward per unit of stake is the same, and equal to $1/4$, in both types
of pools. Therefore this partition is a Nash equilibrium. Every pool at this equilibrium has a total stake equal to 4, which shows that for this instance the PoS
is at most $4/h=4/3$. 
\qed
\end{example}

The previous example also leads to a lower bound on the Price of Stability.

\begin{claim}
\label{cl:lower-bound-oceanic}
    The Price of Stability in the oceanic model, is at least $\frac{4}{3}$.
\end{claim}
\begin{proof}
    We consider a special case of Example \ref{example:non-atomic}, with $h=3$. The set of players consists of only one atomic player with stake $a=2$ and a mass of small players, equal to $m = 4t+2$, where $t$ is a large integer. The optimal partition for this game would be to have one pool including the atomic player and one unit of non-atomic players, and then have the remaining players form pools of size 3, except the last pool which will have one additional unit. Hence $OPT(G) = 1 + 4t/3$ (assuming that $4t$ is divisible by 3). By the construction of Example \ref{example:non-atomic} above, we know that there exists an equilibrium in which a mass of size 2 forms a pool with the large player, and the remainder $4t$ units of small players form $t$ pools of size $4$ each. 
    
    For the sake of contradiction, assume that there is a better equilibrium formation that achieves a ratio better that $\frac{4}{3}$, compared to the optimal partition. In this equilibrium, the large player has to be paired with some mass, say of size $k$, of small players and all other pools have to consist only of small players of mass $l\geq h=3$. The latter pools have to be of the exact same size otherwise the small players would have an incentive to move. In order for this formation to achieve a better than $4/3$ Price of Stability, it has to be the case that $3\leq l<4$. Since, it is an equilibrium it must hold that the large player should not have a motive to move to a pool consisting only of small players, hence by Lemma \ref{lemma:nonatomic-rewards}, we have that $\frac{k-h+a}{k}\geq \frac{a}{l}\Leftrightarrow l(k-h+a)\geq a\cdot k \Rightarrow l(k-1)\geq 2 k \Rightarrow  k>2$, where the implication holds from the fact that $l<4$, $a=2$ and $h=3$. Moreover, no small agent who has been paired with the large player should want to move to a pool of only small players, and vice versa, so that  $\frac{1-\frac{a-h+k}{k}}{k}=\frac{1}{l}\Rightarrow \frac{1}{k^2}=\frac{1}{l}\Leftrightarrow l=k^2$. Combining the facts that $k>2$ and $l=k^2$, we obtain that $l>4$ and thus we arrive to the desired contradiction.

    Therefore, the Price of Stability is lower bounded by
    $\frac{4t/3 + 1}{t+1}$, which tends to 4/3 as $t\rightarrow \infty$.

\end{proof}

Our main result in this section is a matching upper bound, provided that all atomic players have a sufficiently large stake\footnote{The theorem requires large players to hold at least $h/4$ stake, which, coincidentally, is the minimum stake requirement for validators in the Rocketpool platform.} and that the total population of small players is sufficiently large.

\begin{theorem} 
\label{thm:non-atomic}
In the oceanic model, and given any constant $\epsilon>0$, the Price of
Stability for the Shapley scheme, is at most $4/3 + \epsilon$, when all large players have stake in
$(h/4, h)$ and there is a sufficiently large mass of non-atomic players (dependent on $\epsilon$ and the total stake of the large players).
\end{theorem}

The remainder of this subsection is devoted to the proof of Theorem \ref{thm:non-atomic}. Before we proceed, we comment on the type of equilibria that we construct to establish the proof. Following Example \ref{example:non-atomic}, we show that such games always possess equilibria where each large player forms a pool together with some mass of non-atomic players and with no other large players.
We find such equilibria to be natural for the blockchain scenarios that we study 
(certainly more natural than the grand coalition equilibrium from Claim \ref{mycl:PoA}), 
in the sense that big players tend to avoid each other and form pools with less important users. They provide each big player with {\it skin in the game}, as each of them is the only entity with high enough stake in their pool. 
Furthermore, they reduce the question on the Price of Stability to upper bounding the mass of non-atomic players in each pool, as demonstrated in the sequel. 
Finally, we note that some games may also have other equilibria, with more than one large player in some pools, but more equilibria could only positively (if at all) affect the Price of Stability. 

To prove Theorem \ref{thm:non-atomic}, we start with the following lemma, that provides a characterization for the particular type of equilibria we are after.
\begin{lemma}[Equilibrium conditions]
\label{lem:eq-conditions-oceanic}
    Consider a partition of the players into winning pools, where every large player $i\in N_a$ is in a pool with only small players of mass $k_i\leq h$, with $a_i+k_i\geq h$, and all remaining small players are in pools of total mass equal to $l$, with $l\geq h$. Such a partition is at a Nash equilibrium if and only if the
  following conditions are satisfied for any large players $i$, $j$ :
  \begin{align}
    \frac{h-a_i}{k_i^2} &= \frac{1}{l} \label{eq:non-atomic}\\
    \frac{k_i+a_i-h}{k_i} &\geq \frac{a_i}{l} \label{eq:move-to-small} \\
    \frac{k_i+a_i-h}{k_i}  &\geq
      \frac{(h-a_j)^2 - (\max(0, h-a_i-a_j))^2 + (\max(0, a_i-h+k_j))^2}{2k_j^2} \label{eq:move-to-large}
  \end{align}
\end{lemma}

\begin{proof}
We focus first on the non-atomic players.
Since these players are infinitesimally small, a unilateral deviation of any such player to a different pool does not have any effect on the reward per unit of stake in the pool that she moves to.
Therefore, at a Nash equilibrium, all pools with no large player must have the same stake, equal to $l$ in the pools we consider, otherwise the non-atomic players will have an incentive to move to the smallest of these pools. 

We consider now deviations of small players either from a pool with no large players to a pool with a large player 
 or the opposite.
 Constraint \eqref{eq:non-atomic} rules out such deviations.
 Indeed, by
  Lemma~\ref{lemma:nonatomic-rewards} the reward that is left over for the non-atomic players in the pool of player $i\in N_a$ is $1- \frac{a_i-h+k_i}{k_i}$. Since the pool has a mass of size $k_i$ of non-atomic players, the reward per unit of stake for them is equal to $(h-a_i)/k_i^2$. In pools with no large players, the corresponding
  reward per unit of stake is $1/l$. These two rewards must be equal for non-atomic
  players in either pool to have no incentive to switch. 
  
Constraints \eqref{eq:move-to-small} and \eqref{eq:move-to-large}  express the fact that an atomic player $i$ has no reason to switch neither to a pool of non-atomic players nor to the pool of any other atomic player $j$, respectively. 
These follow directly by
  Lemma~\ref{lemma:nonatomic-rewards}. 
\end{proof}

We can now show that it is possible to construct such equilibria for a sufficiently large population of non-atomic players.

\begin{proof}[Proof of Theorem \ref{thm:non-atomic}]
  Consider a game that satisfies the premises of the theorem. We will show that there exists a partition into winning pools that is a Nash equilibrium and in which every pool has at most one large player. We will establish that this equilibrium achieves
  $PoS\leq 4/3$.

  Without loss of generality (simply by scaling), from now on we
  assume that $h=1$. 
  We will show that we can find numbers $k_1,\cdots, k_n$, so that we can fill in the pool of each large player $i\in N_a$, with a mass of non-atomic players, equal to $k_i$, and ensure that the equilibrium constraints of Lemma \ref{lem:eq-conditions-oceanic} are satisfied. 
  Let us also make the simplifying assumption that the
  mass of small players is such that after allocating a mass of size $k_i$ to the pool of each large player $i$, the remaining mass can be partitioned into pools without large players and with
  stake of exactly $l=4/3$. We will revisit this assumption at the end of
  the proof.

  Note that we need $k_i\leq h$ to be able to use Lemma \ref{lem:eq-conditions-oceanic}. The first constraint  (Constraint \eqref{eq:non-atomic}) is an equality and tells us precisely how much $k_i$ should be in terms of $a_i$. By solving for $a_i$, for each $i$, it should necessarily hold that:
  $a_i=h-k_i^2/l=1-3k_i^2/4$. Since by the premises of the theorem,
  $a_i\in(1/4,1)$, we get that $k_i\in(0,1)$, i.e., $k_i\leq h$. 
 An immediate implication
of the assumption that $l\geq 4/3$ is that Constraint \eqref{eq:move-to-small}  is
  satisfied for every $k_i$ because it reduces to
  $(3k_i/4-1/2)^2 \geq 0$.

It remains to show that Constraint \eqref{eq:move-to-large} is also satisfied. 
By substituting $a_i$ and  $a_j$ from Constraint \eqref{eq:non-atomic}, we can see that this last constraint  becomes equivalent to
  $f(k_i,k_j)\geq 0$, where

\begin{align} \label{eq:third-constraint}
f(k_i,k_j) &=  1 - \frac{3}{4}k_i - \frac{9}{32}k_j^2 +
\begin{cases}
     \frac{(3k_i^2 + 3k_j^2 - 4)^2}{32k_j^2} & \text{if } 3k_i^2 + 3k_j^2 \geq 4 \\
    0 & \text{otherwise}
\end{cases} 
- \begin{cases}
    \frac{(3k_i^2 - 4k_j)^2}{32k_j^2} & \text{if } 3k_i^2 \leq 4k_j \\
    0 & \text{otherwise}
\end{cases} 
\end{align}

  \begin{figure}
    \centering
    \includegraphics[scale=0.5]{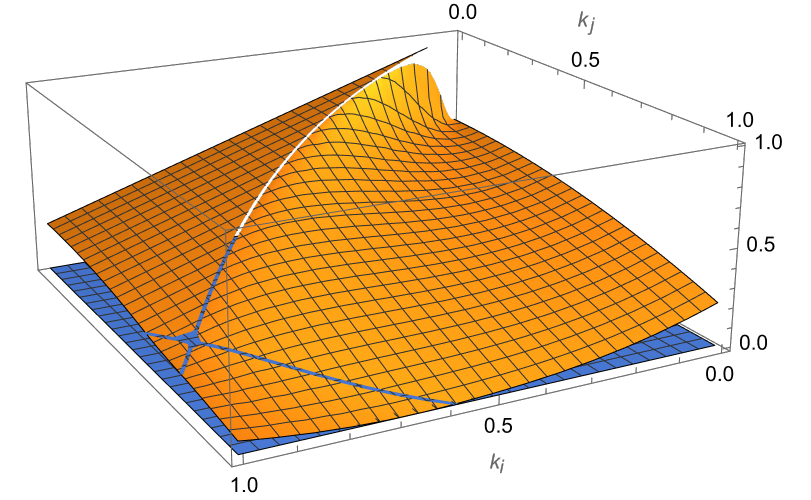}
    \caption{\footnotesize The function $f(k_i,k_j)$ of Equation~\eqref{eq:third-constraint}. 
    For $k_i,k_j\in[0,1]$, the minimum value is 0 and it is achieved 
    when $k_i=2/3$ and $k_j=1$.}
    \label{fig:third-constraint}
  \end{figure}

The following technical claim can then be established.
\begin{claim}
\label{cl:f(ki,kj)}
    For any $k_i, k_j\in (0, 1]$, we have $f(k_i, k_j)\geq 0$.
\end{claim}

Figure~\ref{fig:third-constraint} shows a plot of $f(k_i, k_j)$,
  which demonstrates that it is non-negative for all
  $k_i, k_j \in [0, 1]$. One can prove analytically Claim \ref{cl:f(ki,kj)}, using standard arguments and an exhaustive case analysis;
  we defer this to the full version of the paper.

So far, we have established that we can have a Nash equilibrium, with the values of $k_i$ set by \eqref{eq:non-atomic}. Under such a partition, the pools without a large player have size $l = 4/3$, while pools with one large player have size 
  $k_i + a_i = 4/3 - (2 - 3k_i)^2/12 \leq 4/3$. Since the optimal partition, at its best would split the players into pools of total stake equal to $h=1$, the Price 
  of Stability is at most $4/3$.
  A matching lower bound is given 
  by Example~\ref{example:non-atomic} and Claim \ref{cl:lower-bound-oceanic}.

  This completes the proof of the theorem when all pools of non-atomic
  players have stake of exactly $l=4/3$. If the total stake is such
  that this is not possible, we can construct a Nash equilibrium with
  pools of stake $l=4/3+\epsilon$, where $\epsilon$ can be made arbitrarily small, provided there is a sufficiently large mass of
  non-atomic players. The preceding argument applies almost
  identically to this case as well. 
  This
  complicates some expressions, regarding the function $f(k_i, k_j)$, but it does not affect the core logic
  of the proof.
\end{proof}

An interesting consequence of the theorem's proof is that large players 
receive more reward at the Nash equilibrium than their proportional 
share. Specifically, their reward is $(a_i - h + k_i)/k_i$, while their 
proportional share is $a_i/(a_i + k_i)$. Given that 
$a_i = 1-3k_i^2/4$, we get the following.
\begin{fact}
  In the Nash equilibria of Theorem~\ref{thm:non-atomic}, when $h=1$, the reward of a large player in a pool of mass of non-atomic players $k_i$ is more than its proportional sharing by 
  \[
    \frac{k_i (3 k_i - 2)^2}{4 (2 - k_i) (3 k_i + 2)} \geq 0,
  \]
\end{fact}
The extra reward varies between $0$ (when $a_i=2/3$) and $0.05$ (when $a_i\approx 1/4)$. 

\subsection{The Price of Stability in the atomic model}

In this subsection, we provide an analogous analysis for the atomic case, where we could have again large and small players, but not infinitesimally small, i.e., $N_s=\emptyset$, and therefore $N = N_a$. We believe it is valuable to understand the performance under the purely atomic model as well. This provides a more complete picture for these games and intends to answer the question of whether a low Price of Stability is maintained even without infinitesimally small players. When there are only atomic players, the discrete nature of the problem makes the Price of Stability go up, albeit not by a lot in the cases we consider below. This is consistent with the analogous flavor of results in the analysis of congestion games with atomic and non-atomic players \cite{roughgarden16}. 

We start with the following easy observation, showing that we have a worse lower bound when there is no mass of non-atomic players. 

\begin{claim}
\label{cl:lower}
The analysis of Example \ref{example-PoS} yields that in the atomic model, $PoS \geq 2$.
\end{claim}

\noindent {\bf Assumption on the stake distribution.}
The analysis for the atomic model quickly becomes much more technical and challenging for an arbitrary stake distribution, due to the expression for the Shapley formula. 
But following the same spirit as with the oceanic model, we have managed to analyze a class of games that follow a separation of $N$ into large and small players.  
In particular, we will consider a 2-valued stake distribution, so that for every player $i\in N$, either $a_i=1$ modeling a small player or $a_i = a>1$, with $a<h$, modeling a large player. We leave as an open problem the analysis for the case where the large players can have different stakes.

Let $n_b, n_s$ be the number of large and small players respectively, so that $n=n_s+ n_b$. Our main result of this subsection is the following theorem, where we give a matching upper bound to Claim \ref{cl:lower}, as long as the number of small players is sufficiently large, so that again the Price of Stability is bounded by a small constant.

\begin{theorem}
\label{thm:pos-general}
    For instances on $n=n_s+n_b$ players, with $n_s\geq (2h-1)^2+n_b\cdot(h+1)$, it holds that  PoS $ \leq 2$. 
\end{theorem}
To prove this, we will construct an equilibrium, with a similar pattern as in the proof of Theorem \ref{thm:non-atomic}, where each large player forms a pool with only small players. In the atomic model, such partitions are captured by the following definition.

\begin{definition}
\label{def:k-l}
For 2 integers $k, l$, with $k+a\geq h$ and $l\geq h+1$, we say that a partition of $N$ is a $(k, l)$-partition %
if it consists only of pools of the following form (without necessarily having pools from all the three types below) 
    \begin{itemize}
        \setlength{\itemsep}{1pt}
        \setlength{\parskip}{0pt}
        \setlength{\parsep}{0pt}
        \item Pools with exactly one big player and $k$ small players (Type A pools),
        \item Pools with $l$ small players (Type B pools),
        \item Pools with $l-1$ small players (Type C pools).
    \end{itemize}
\end{definition}

A $(k, l)$-partition that is also a Nash equilibrium is referred to as a $(k,l)$-equilibrium. 
The reason we may need both Type B and Type C pools is that the total number of agents assigned to these pools may not be divisible by $l$ or $l-1$. 
Furthermore, if we find equilibria where $k$ and $l$ are upper bounded by a small multiple of $h$, this yields directly a bound on the Price of Stability.

\subsubsection{Games with one large player}
\label{subsec:shapley-1}

It is very instructive (and essentially the core of the proof) to provide first the analysis when we have only one large player with stake $a$ (we fix this to be player $1$). All the other players have stake equal to $1$. As we will see, this is already a non-trivial case.

We first provide some useful properties that help us understand the Shapley value scheme as well as the structure of its equilibria. The next lemma is the analog of Lemma \ref{lemma:nonatomic-rewards} and is derived by applying the Shapley formula \eqref{eq:shapley}.

\begin{restatable}{lemma}{Lemma3.7}
\label{lem:typeA}
Consider a winning pool $S$, containing the large player and $k$ small players.
\begin{itemize}[noitemsep]
    \item If $k< h$, then the Shapley scheme provides to player $1$ a payoff of $\phi_1(S)= \frac{k-h+a+1}{k+1}$. Each of the other players receives a payoff of $\frac{h-a}{k\cdot(k+1)}$.
    \item If $k\geq h$, then player $1$ has a payoff of $\phi_1(S)= \frac{a}{k+1}$. Each of the other players receives a payoff of $\frac{k+1-a}{k\cdot(k+1)}$.
\end{itemize}
\end{restatable}

The following is straightforward by applying \eqref{eq:shapley}.

\begin{fact}[Pools of Type B or C]
\label{fact:small_pool}
    Consider a winning pool $S$ consisting of only small players, with $|S| = t$, for some $t\geq h$. Then for every $i\in S$, $\phi_i(S) = 1/t$.
\end{fact}

The next lemma justifies why in Definition \ref{def:k-l}, it suffices to consider pools of Type B and Type C, regarding the pools containing only small players, in the construction of equilibria. 

\begin{lemma}
\label{lem: 4.4}
    Suppose there exists a Nash equilibrium containing in its partition 2 pools $S_1, S_2$, that consist only of small players. Then, it must hold that 
    $$ | |S_1| - |S_2| | \leq 1$$
\end{lemma}

Therefore, for games with a single large player, when $a<h$, the equilibria of the game must have the format of a $(k, l)$-partition, for appropriate values of $k$ and $l$. Each of the pools should have total stake that exceeds the threshold $h$, so that they are winning pools, and this is why in Definition \ref{def:k-l}, we enforced that  $k+a\geq h$ and $l-1\geq h$ ($l\geq h$ some times also suffices if there are no pools of Type C). 
The next step is to understand how to set the parameters $k$ and $l$ so that a $(k, l)$-partition is a Nash equilibrium. The following lemma describes a set of sufficient conditions, which are also necessary when the partition has pools of all types (A, B and C). 
For our purposes, we will stick below to the case that $k\leq  h-1$.

\begin{lemma}[Sufficient conditions]
\label{lem:NE_conditions}
    Consider a $(k, l)$-partition with $k<h$. When $k\leq  h-2$, the partition is a Nash equilibrium, if the following conditions hold:
    \begin{enumerate}[noitemsep]
        \item $\frac{k(k+1)}{h-a}\leq l\leq\frac{(k+1)(k+2)}{h-a}$
        \item $l\geq \frac{a(k+1)}{k-h+a+1}$
    \end{enumerate}

\noindent When $k=h-1$, the only change is that the upper bound in the first condition needs to be replaced by: $l \leq \frac{(k+1)(k+2)}{h+1-a}=\frac{h(h+1)}{h+1-a}$.   
\end{lemma}

Based on the above, we prove below that for large enough games, there always exist Nash equilibria (characterized by the values of $k$ and $l$), which are relatively good compared to the optimal formation.

\begin{theorem}
\label{thm:PoS}
    For instances with a single large player and with
    $n\geq (2h-1)^2+h$,  PoS $\leq 2$.
\end{theorem}
\begin{proof}
    We provide an outline and defer to the Appendix the proofs of Lemmas \ref{lem:k,l values-part1} and \ref{lem:k,l values-part2}, that are stated below, which are the most technical parts of the entire proof. We show that there exists a $(k, l)$-equilibrium profile, 
    with $l\leq 2h$ and $k+a\leq2h$. This would mean that all the pools at the equilibrium have stake at most $2h$. Since the optimal formation has pools with stake at least $h$, the Price of Stability is at most 2. %

    To ensure the existence of such an equilibrium, we exploit the sufficient conditions of Lemma \ref{lem:NE_conditions}, and we distinguish 2 cases, based on whether we will search for $k\leq h-2$ or for $k=h-1$. 

    \noindent {\bf Case 1: $h\leq a^2-2a+2$.}\\
    In this case we will identify a value for $k$ so that $k\leq h-2$. The first condition of Lemma \ref{lem:NE_conditions} defines an interval $\Delta$ for the possible values of $l$, which contains at least two integers. This is true since $\frac{(k+1)(k+2)}{h-a}-\frac{k(k+1)}{h-a}=\frac{2(k+1)}{h-a}\geq 2$, for any feasible $k$, since we are only looking for values of $k$ with $k +a \geq h$. Hence, if the interval $\Delta$ belongs entirely to the region for $l$ defined by the second condition of Lemma \ref{lem:NE_conditions}, there exists an equilibrium where both $l-1$ and $l$ can take an integer value belonging to $\Delta$. In order for this to happen, 
    it must hold that: $\frac{k(k+1)}{h-a}\geq \frac{a(k+1)}{k-h+a+1}\Leftrightarrow k^2+k(a-h+1)-a(h-a)\geq 0$.

    We can proceed with solving the inequality of degree two, with respect to $k$. It is easy to check that the discriminant is positive. Hence, by carrying out the calculations, we eventually get that the inequality holds for the integer values that satisfy 
     \begin{align}
     \label{k*_value}
     k\geq k^* =  \frac{1}{2}(\sqrt{-3a^2+h^2+2ah-2h+2a+1}-a+h-1).
    \end{align}
The remainder of the proof for this case now is to show how to set $k$ and $l$ appropriately. 
  \begin{lemma}
  \label{lem:k,l values-part1}
      For any pair of integer values of $a$ and $h$ for which $2\leq a\leq h-1$, and $h\leq a^2-2a+2$, there are integer values $k,l$ such that: $k\geq k^*$, $h-a\leq k\leq h-2$ and $h+1\leq l\leq 2h$.
  \end{lemma}

\noindent {\bf Case 2: $h > a^2-2a+2$.}\\
Now we cannot guarantee that $k\leq h-2$ and we need to follow a different (and in fact simpler) approach. We have again an interval $\Delta$ for the values of $l$, but as we will be looking for $k=h-1$, we need to use the analogous conditions from Lemma \ref{lem:NE_conditions}. The remaining part for handling Case 2 is by the lemma below. 

\begin{lemma}
\label{lem:k,l values-part2}
    For any pair of integer values of $a$ and $h$ for which it holds that $2\leq a\leq h-1$, $h> a^2-2a+2$, there are integer values $k,l$ such that: $h-a\leq k\leq h-1$ and $h+1\leq l\leq 2h$.
\end{lemma}

Given Lemmas \ref{lem:k,l values-part1} and \ref{lem:k,l values-part2}, we can now describe the equilibrium construction for any $n\geq l^2+ k  +1$ (for the integer values $k,l$ that were determined in Lemmas \ref{lem:k,l values-part1} and \ref{lem:k,l values-part2}), and thus for any $n\geq (2h-1)^2+2h$. Since the values for $k$ and $l$ that we found satisfy the equilibrium conditions, all that remains to show is that we can actually partition the $n$ players into such a $(k, l)$-partition. To do this, we will have $k+1$ players, including the large player, form one Type A pool. For the remaining $n-k-1\geq (l-1)^2$ players, we first create as many Type C pools with exactly $l-1$ players as possible. And then, each of the remaining $(n-k-1)\mod (l-1)$ players is added to a different Type C pool, converting it to a Type B pool. This completes the proof of Theorem \ref{thm:PoS}.
\end{proof}

\subsubsection{Games with more large players}
\label{subsec:shapley-more}

Building upon the results of the previous subsection, we can address the more general case, where there are more large players. Suppose we have $n_b$ large players and $n_s$ small players so that $n=n_b+n_s$. As before, we show that there is again a $(k, l)$-partition that is an equilibrium and so that each pool has stake at most $2h$.
To prove this, we need again to identify sufficient equilibrium conditions. But for this we can still utilize Lemma \ref{lem:NE_conditions} with one additional condition: that large players have no incentive to deviate to another Type A pool. 

The remaining argument for completing the proof of Theorem \ref{thm:pos-general} is in Appendix \ref{app-sec:shapley}.

\subsubsection{Lower bounds for the atomic model with enough small players}

For the atomic model, Claim \ref{cl:lower} shows that Example \ref{example-PoS} provides a lower bound of 2. However the example does not satisfy the premises of Theorem \ref{thm:pos-general}, that we have sufficiently many small players of stake 1. Under this assumption the best lower bound we have established is 1.5, as stated below. Therefore, for the atomic model, it remains an open problem whether the Price of Stability is lower than 2 when there is a sufficiently high population of small players or whether Theorem \ref{thm:pos-general} is tight.

\begin{theorem}
\label{thm:atomic lower bound}
    There exist instances where the number of small players satisfies the assumption of Theorem \ref{thm:pos-general}, and where $PoS\geq 1.5$.
\end{theorem}
\begin{proof}
    In order to prove this lower bound of $\frac{3}{2}$ we will consider a simple instance, where there is only one large player with stake $a=2$, the threshold is equal to $h=4$, and there is a large amount of small players with stake 1. In this instance, since there is only one "big" player and many small players, the best equilibrium is necessarily in the form of a $(k,l)$-equilibrium, with the lowest possible value for $l$. 
From the conditions of Lemma \ref{lem:NE_conditions}, we have that in a $(k,l)$-equilibrium it must hold that $\frac{k(k+1)}{h-a}\leq l\leq \frac{(k+1)(k+2)}{h-a}$ and $l\geq\frac{a(k+1)}{k-h+a+1}$. In order to find the best equilibrium, it suffices to find the smallest value of $k$, for which the two above conditions overlap. This happens for $k=2$ and the two conditions above become $3\leq l \leq 6$ and $l\geq 6$. Hence, the only feasible value for $l$ to have a $(k,l)$-equilibrium is $l=6=\frac{3h}{2}$. Given that the optimal formation consists of pools with total stake at least equal to $h$, this leads us to the lower bound of $\frac{3}{2}$ for the Price of Stability.

\end{proof}
\section{Optimality of the proportional sharing scheme}
\label{sec:proportional}

We come now to compare the Shapley scheme with the simplest possible scheme that one could apply in this setting, namely proportional sharing.

Recall that for a pool $C$, we have denoted by $m(C)$ the total stake of the pool, including atomic and non-atomic stake. For atomic players, the proportional share of an agent $i$ is the percentage of her contribution over $m(C)$, where $C$ is the pool that she belongs to. For the non-atomic players, the reward per unit of stake is equal to the profit divided by the total stake of the pool. This is a quite standard and popular rule, that is being used widely in several off-chain and on-chain agreements.
\begin{definition}
    For a player $i\in N_a$, with stake $a_i$ belonging to a pool $C$, the reward she receives by the proportional reward sharing scheme is $p_i(a_i, C)=\frac{a_i}{m(C)}\cdot \rho(C).$ For non-atomic players, the reward per unit of stake is $\rho(C)/m(C)$. 
\end{definition}
For this scheme, there is always an equilibrium that induces an optimal formation, and this holds for any stake distribution, without the premises of Theorem \ref{thm:non-atomic}.

\begin{theorem}
\label{thm:PoS_proportional}
    In the oceanic model, the proportional sharing scheme has Price of Stability equal to 1. This holds independently of the volume of non-atomic players (and hence it holds for the purely atomic model as well).
\end{theorem}

\begin{proof}
Consider an optimal partition of the players into pools, as defined in \eqref{eq:OPT}. If there are multiple optimal partitions, consider the one that is as balanced as possible, i.e., a lexicographically optimal (leximin) one, where the stake of the smallest pool is maximized, and subject to that, the stake of the second smallest pool is maximum, and so on.  Let $C_1,\dots, C_t$ be the pools of this partition. It holds that $m(C_j)\geq h$ for any $j\in [t]$. 

We note first that all the non-atomic players must belong to the pool (or pools) of minimum stake. If not, we could transfer some stake to the smallest pools and arrive at a lexicographically better one. 
Hence, for all non-atomic players, the reward per unit of stake is $\frac{1}{\min_j m(C_j)}$. This means that they have no incentive to switch to any other pool, where the reward per unit of stake is lower.
Let us consider now the atomic players. If the partition is not an equilibrium, then there exists a player $i\in N_a$, belonging to some pool $C_k$, who becomes better off by moving to a pool $C_r$. But then this means that $\frac{s_i}{m(C_k)} < \frac{s_i}{s_i + m(C_r)}$.
    This implies that $m(C_k) > s_i + m(C_r)$, i.e., that $m(C_k\setminus \{i\}) > m(C_r)$. But then $i$ is included in a pool $C_k$ in the optimal partition that we considered, where even without $i$, $C_k$ has more stake than $C_r$. This means that the partition we started with was not a lexicographically optimal partition, a contradiction.
\end{proof}

\section{Sybil resilience}
\label{sec:Sybil}

So far, it may appear that the proportional scheme is a dominating solution within the space of candidate reward schemes. In this section, we highlight one drawback of proportional sharing, which provides an advantage of Shapley, related to 
the resilience to Sybil strategies. By a Sybil strategy we refer to the act of an atomic player to split her stake into smaller portions so as to participate 
in different pools, 
disguised as a set of different players 
(which can be feasible, given the anonymity in blockchain environments). 

\subsection{Sybil resilience of equilibria}
\label{subsec:Sybil-Resiliency non-atomic}

There are different ways to define Sybil attacks in the literature, depending on the model at hand and on the allowed or feasible behaviors from the players. 
For our model, we consider Sybil attacks with respect to some given configuration (which may have formed from previous actions that the players took).
We also consider that only the atomic players are eligible to perform a Sybil attack.

\begin{definition}
\label{def:sybil-attack}
    Fix a strategy profile inducing a partition $\Pi$, of the players into pools. Given $\Pi$, a Sybil strategy for a player $i\in N_a$ with stake $a_i$ is a strategy where the player splits her stake into $s_{i_1},\dots, s_{i_t}$ for some $t$, with $\sum_{j=1}^t s_{i_j} = a_i$ and chooses $t$ pools to join among the ones formed by $\Pi$, contributing stake $s_{i_j}$ to pool $j$, for $j\in [t]$. The payoff under such a strategy for player $i$ is the sum of the payoffs from all the $t$ pools.       
\end{definition}
Based on the above definition of a Sybil attack, below we define the notion of Sybil-proofness (alternatively Sybil-resiliency).
\begin{definition}
\label{def:sybil-at-profile}
    Fix a partition $\Pi$. We say that a reward sharing scheme is Sybil-proof w.r.t. the partition $\Pi$ if no player $i$ can become strictly better off by switching to a Sybil strategy. 
\end{definition}

Ideally, we would like to have reward schemes where no player has an incentive for Sybil attacks regardless of the configuration the game might be at. 
This is too much to ask for, but we start here with a positive result, that concerns Sybil resilience w.r.t. the equilibrium formations of Section \ref{sec:shapley}. In particular, for the Shapley scheme, we show that even after enlarging the strategy space with Sybil strategies, 
the equilibrium partitions described in Section \ref{sec:shapley}, continue to be at equilibrium. \begin{theorem}
\label{thm:Sybil_shapley non-atomic}
    In the oceanic model, the Shapley value reward sharing scheme is Sybil-proof w.r.t. the equilibria identified in Theorem \ref{thm:non-atomic}, and more generally w.r.t. any equilibrium as described by the conditions of Lemma \ref{lem:eq-conditions-oceanic}.
\end{theorem}
\begin{proof}
    Consider an equilibrium formation satisfying the conditions of Lemma \ref{lem:eq-conditions-oceanic}, with $d$ pools $C_1,\dots, C_d$. Each pool contains either an atomic player $j$ with stake  $a_j$ with a mass of $k_j$ small players (such that $a_j+k_j\geq h)$, or a mass of $l>h$ small players.
     
     Suppose that an atomic player $i$ with stake $a_i$ decides to split her stake into $m$ different portions, $s_{i_1},s_{i_2},...,s_{i_m}$ and commit $s_{i_t}$ to pool $C_t$, where $\sum_{t\in [m]} s_{i_t} = a_i$.
     We distinguish the following cases.

  \noindent {\bf Case 1:} $C_t$ is the pool that $i$ belonged to initially. Then by Lemma \ref{lemma:nonatomic-rewards}, her reward is $\frac{k_i-h+s_{i_t}}{k_i}$.

  \noindent {\bf Case 2:} $C_t$ is a pool, with a mass of $l\geq h$ small players. By contributing $s_{i_t}$ to such a pool, Lemma \ref{lemma:nonatomic-rewards} shows that player $i$ receives a reward equal to $\frac{s_{i_t}}{l}$. This is at most $\frac{s_{i_t}}{k_i}$, due to the fact that $k_i<h=1\leq l$.

    \noindent {\bf Case 3:} $C_t$ is a pool with another atomic player $j$ with stake $a_j$ and a mass of small players $k_j$. 
    For ease of notation in this case, let $s_{i_t} = b$ and without loss of generality, let $h=1$. We will establish that the reward player $i$ gets is at most $\frac{b}{k_i}$. 
    After the Sybil attack of player $i$, $C_t$ contains player $j$ with stake $a_j$, player $i$ with stake $b$, and a mass of size $k_j$ of small players. We need to discriminate between some subcases here, since the formula of the reward agent $i$ receives varies. The reward agent $i$ receives, by Lemma \ref{lemma:nonatomic-rewards} is $\frac{(1-a_j)^2 - \max(0, 1-b-a_j)^2 + \max(0, b-1+k_j)^2}{2{k_j}^2}$. 

    Since we are in an equilibrium formation, from the conditions in Lemma \ref{lem:eq-conditions-oceanic} it holds that $k_i=\sqrt{l\cdot (1-a_i)}$ and $k_j=\sqrt{l\cdot (1-a_j)}$. Also, since $k_i, k_j\leq h$, it holds that $a_i,a_j\in [1-\frac{1}{l},1)$. Note also that the function $\frac{b}{k_i}=\frac{b}{\sqrt{l\cdot (1-a_i)}}$, is increasing with respect to $a_i$. At the same time, for the reward function of player $i$ within pool $C_t$, we have the following property, which can be easily verified:
    \begin{claim}
       The function $\frac{(1-a_j)^2 - \max(0, 1-b-a_j)^2 + \max(0, b-1+k_j)^2}{2{k_j}^2}$, is decreasing with respect to $a_j$, for $a_j\in[0,1)$. 
    \end{claim}
    \noindent So in order to show that $\frac{b}{k_i}\geq\frac{(1-a_j)^2 - \max(0, 1-b-a_j)^2 + \max(0, b-1+k_j)^2}{2{k_j}^2} $, it suffices to show that this holds for the lowest values that $a_i$ and $a_j$ can take, which is $1-\frac{1}{l}$.
    
    When $a_i=a_j=1-\frac{1}{l}$, it holds that $k_i=k_j=1$. Given the form of the right hand side, we consider first the subcase that $1-b-a_j<0\Rightarrow b>1-a_j=1-(1-\frac{1}{l})=\frac{1}{l}$. Also note that $b-1+k_j\geq 0$, since $k_j=1$. In this scenario, the Sybil attacker gets reward from $C_t$ equal to $\frac{(1-a_j)^2 + (b-1+k_j)^2}{2{k_j}^2}$. The desired inequality is equivalent to showing that $2b\geq (1-(1-\frac{1}{l}))^2+b^2=\frac{1}{l^2}+b^2$, which holds due to the fact that $b>\frac{1}{l}\Rightarrow b>\frac{1}{l^2}$, for any $l\geq 1$ and also $b\geq b^2$, since $b\leq a_i<1$. 

    Next, we move to the subcase where $1-b-a_j\geq 0 \Rightarrow b\leq 1-a_j=\frac{1}{l}$. Again, $k_i=k_j=1$, so the desired inequality becomes 
    $2b\geq (1-(1-\frac{1}{l}))^2-(\frac{1}{l}-b)^2+b^2=\frac{2b}{l}\Rightarrow 2b\geq \frac{2b}{l}$, which holds for any $l\geq 1$.

    Putting everything together, the initial reward player $i$ receives is $\frac{k_i-h+a_i}{k_i}$. Let $C_1$ be the pool where $i$ belonged to initially, before the Sybil attack.
    We can then rewrite this as follows
    $$\frac{k_i-h+a_i}{k_i} =\frac{k_i-h+(\sum_{t=1}^m s_{i_t})}{k_i}=\frac{k_i-h+s_{i_1}}{k_i}+\sum_{t=2}^m \frac{s_{i_t}}{k_i}$$
    But now we can see that each term in the summation is at least as big as the reward received by each pool $C_t$ analyzed in Cases 2 and 3. Also the leftmost term is at least as good as the reward received at Case 1. Hence player $i$ earns more rewards by staking all her stake in her original pool.

\end{proof}

Next, we analyze the proportional scheme. In contrast to the Shapley value, and to our surprise, the proportional scheme has a disadvantage here. We show that the optimal equilibrium formation identified in Theorem \ref{thm:PoS_proportional} is not always Sybil-proof. %

\begin{theorem}
\label{thm:Sybil_proportional}
    There exist instances where the proportional reward sharing scheme is not Sybil-proof w.r.t. the optimal equilibrium of Theorem \ref{thm:PoS_proportional}. 
\end{theorem}
\begin{proof}
    We show first a simple example when there is no ocean of small players. Consider an instance with 3 large players, who all have stake $a=h/2+1$, and with $h/2-1$ additional players, with stake equal to $1$. Obviously $OPT(G) = 2$, where the optimal formation is to have 2 pools $C_1, C_2$, with $C_1$ containing two large players together and $C_2$ containing the remaining players. This is an equilibrium, as follows by Theorem \ref{thm:PoS_proportional}, and the payoff for each player in $C_1$ is $1/2$. Suppose now that one of the large players in $C_1$ decides to allocate 1 unit of stake to $C_2$ and keep the remaining stake in $C_1$. Her total payoff from this Sybil attack will be equal to $\frac{h/2}{h+1} + \frac{1}{h+1} > 1/2$.
    
    Even when we have a large mass of non-atomic players, one can still have analogous constructions. The reason is that when the total stake of the game is not a multiple of $h$, the scheme then becomes vulnerable to Sybil attacks.
    To illustrate this, consider an instance with $n$ atomic players and a mass of small players, such that the total stake is not a multiple of $h$. This means that in the optimal equilibrium formation of Theorem \ref{thm:PoS_proportional}, the pool of maximum stake, say $C_1$ will have total stake $\ell_1> h$. Suppose that $C_2$ is the pool of minimum stake, equal to $\ell_2$, with $h\leq \ell_2 \leq \ell_1$.  
    Let $i$ be an atomic player in $C_1$ with stake $a_i$. There exists such a player, since by the proof of Theorem \ref{thm:PoS_proportional}, we know that when $\ell_1> \ell_2$, then $C_1$ has to contain atomic players (if $\ell_1 = \ell_2$, then we pick as $C_1$ any pool that contains an atomic player). We argue that agent $i$ can gain more total reward by splitting her stake into two portions, specifically $a_i-x$ which she will commit to her initial pool $C_1$ and $x$ (s.t. $\ell_1-x\geq h$, which is feasible if $x$ is small enough), which she will commit to pool $C_2$. By this splitting, agent $i$ will receive a total reward of $\frac{a_i-x}{\ell_1-x}+\frac{x}{\ell_2+x}$, based on the proportional reward sharing scheme. Suppose now that $\ell_1 = \ell_2 + \delta$, for some $\delta\geq 0$. By simple calculations, we can see that the reward of $i$ under this Sybil attack is greater than her initial reward of  $\frac{a_i}{\ell_1}$, as long as 
    $$x\leq \frac{\ell_2(a_i + \delta) + \delta^2}{2\ell_2 + 2\delta -a_i}$$. 
\end{proof}

The next definition captures the stronger property that one could envision, i.e., Sybil resilience no matter what is the initial configuration at hand (equilibrium or not). %

\begin{definition}
\label{def:sybil}
    A reward sharing scheme is Sybil-proof if it is Sybil-proof w.r.t. any partition $\Pi$.
\end{definition}

This is obviously a quite demanding requirement, especially since some partitions may be very unlikely to form in practice. Theorem \ref{thm:Sybil_proportional} shows already that proportional sharing does not satisfy it. For the Shapley scheme we can also construct non-equilibrium profiles where Definition \ref{def:sybil} fails. To see this consider a partition $(C_1, C_2)$, where $C_1$ contains 2 atomic players with stake $a>h/2$ and $C_2$ contains non-atomic players of mass $l=h$. Then it is easy to verify, using Lemma \ref{lemma:nonatomic-rewards}, that a large player from $C_1$ can keep a stake of $h/2$ in the first pool and contribute the remaining to $C_2$ resulting in a higher payoff for her. 

Hence, when there is no significant cost of carrying out a Sybil attack, we cannot expect
to avoid such strategies from any configuration $\Pi$, except with trivial schemes (e.g., zero payments). In fact such impossibilities have already been pointed out under a quite different model in \cite{kwon2019impossibility}. We exhibit below that for our model, this holds beyond the Shapley and the proportional scheme, for a wide class of reward schemes.

\begin{observation}
\label{thm:non-Sybil}
    Consider a reward scheme where the payment $p_i(a_i, S)$ of a player with stake $a_i$ in a pool $S$, depends on $a_i$ and on the stake distribution of $S_{-i}$ but in an anonymous way. Suppose further that the payment is concave  w.r.t. $a_i$, even in a subset of the entire range of $(0, h)$ (with at least one value of $a_i$ where concavity holds with strict inequality). Then the scheme is not Sybil-proof.
\end{observation}
\begin{proof}
     Consider a partition where player $i$ ends up in a pool $S$, and where two other pools have formed $S_1, S_2$, such that $m(S_1) = m(S_2) = m(S\setminus\{i\})$ and the stake distribution of $S_1$ and $S_2$ is the same with the distribution of $S\setminus \{i\}$ (meaning that the composition in terms of the stakes of the atomic players and the volume of the non-atomic players is the same). Then suppose that player $i$ splits her stake $a_i$ into 2 portions of $a_i/2$ and commits it to $S_1$ and $S_2$. Then by the assumption of the theorem, regarding strict concavity, we would have that $p_i(a_i, S) < p_i(a_i/2, S_1) + p_i(a_i/2, S_2)$.

\end{proof}

\subsection{Computation of Sybil strategies}

We conclude this section by commenting on the computational problem of actually finding a good Sybil strategy.
For this, fix a player $i$ with stake $a_i$, and an arbitrary strategy profile of the other players $\vec{x}_{-i}$.
Imagine that player $i$ enters the game at a time where the strategies of the other players induce a partition $\Pi_{-i} = (C_1,\dots, C_m)$, where each $C_j$ is a winning pool (we ignore non-winning pools, if there are any, under the current partition).
If player $i$ wishes to maximize her utility, the problem that she has to solve is to compute how to split her stake and commit it to different pools in order to receive the maximum  possible amount of reward. More formally, let $s_{i_1},s_{i_2},...,s_{i_m}$, be the variables representing the portions of stake of $a_i$, that agent $i$ will assign to each of the $m$ pools. 
Then player $i$ faces the following optimization problem:

\setlength{\belowdisplayskip}{1pt} \setlength{\belowdisplayshortskip}{1pt}
\setlength{\abovedisplayskip}{1pt} \setlength{\abovedisplayshortskip}{1pt}
\begin{equation*}
\begin{aligned}
& {\text{maximize}}
& & \sum\limits_{j=1}^m p_i(s_{i_j}, C_j\cup \{i\}) \\
& \text{subject to}
& & \sum\limits_{j=1}^m s_{i_j}\leq a_i \\
& & &   s_{i_j} \geq 0 ~~\forall j\in [m]%
\end{aligned}
\end{equation*}

The constraints above are linear but this is not always an easy problem to solve, and its complexity is dependent on the reward scheme, since this affects the form of the objective function. 
For simpler schemes, such as the proportional one, we can have a polynomial algorithm.
In particular, under the proportional scheme, the objective function is a sum of concave functions (and hence a concave function itself), equal to $\sum_{j=1}^m \frac{s_{i_j}}{m(C_j) + s_{i_j}}$. This means that the optimization problem is equivalent to minimizing a convex function subject to linear constraints. In fact, for the proportional scheme, we can have a much simpler, greedy algorithm, based on a water-filling argument (see Appendix \ref{proof-water-filling}). Hence, we can conclude with the following.
\begin{theorem}
\label{thm:water-filling}
    Under the proportional scheme, 
    the optimal solution of the above problem, can be computed in polynomial time.
\end{theorem}

\noindent{\bf Proportional vs Shapley w.r.t. complexity of Sybil attacks.} The above algorithm shows that for the proportional scheme, it is relatively easy to find a profitable Sybil attack (whenever one exists). 
On the other hand, 
this is not known to be applicable for the Shapley scheme.
For the restriction to integer stakes, solving the aforementioned optimization problem is PP-hard \cite{rey2014false}. Without integrality constraints, we are not aware of any efficient algorithm. Even if we have access to an approximation algorithm for the Shapley value of a pool, it does not seem obvious how to find an approximately optimal Sybil strategy. 
Hence the overall computational difficulty of the Shapley value could serve as a deterrent for potential manipulators.

\section{Alternative variations of proportional schemes}
\label{sec:alt}

We conclude our study with some further ideas that have been proposed in the literature. 
In \cite{chen2019axiomatic}, some alternative schemes are considered, based on variations of proportional sharing. This involves functions that are superadditive in a player's stake. For illustration, we examine the proportional-to-squares reward sharing scheme.
Qualitatively, similar results hold if we use other functions of the stake in the form $a_i^\beta$ with $\beta>1$.

\begin{definition}
    For a player $i$, with stake $a_i$ belonging to a pool $C$, the reward by the proportional-to-squares reward sharing scheme is $p_i(a_i,C)=\frac{a_i^2}{\sum\limits_{j\in C} a_j^2} \cdot \rho(C)$.
\end{definition}

On the opposite direction, we could have a subadditive function of each player's stake.

\begin{definition}
    For a player $i$, with stake $a_i$ belonging to a pool $C$, the reward she receives by the proportional-to-square-roots 
    scheme  is $p_i(a_i, C)=\frac{\sqrt{a_i}}{\sum\limits_{j\in C}\sqrt{a_j}} \cdot \rho(C)$.
\end{definition}

We note that we have presented the definition above only for the purely atomic model, as taking the square or the square root of infinitesimally small quantities leads to uninteresting or trivialized payment rules in the limit. E.g., for a pool with one atomic player and a mass of non-atomic players, the proportional to squares rule would allocate the entire reward to the atomic player.

Even though such schemes can have their merits, we argue here that in our context they also have their drawbacks related to the metrics of the previous sections.
We first exhibit that the proportional-to-squares scheme can lead to very bad equilibria, attaining very high Price of Stability. 

\begin{theorem}
\label{thm:sq_lb}
   For any $h$, there exist instances where the proportional-to-squares scheme has Price of Stability at least $\frac{h-1}{2} = \Omega(h)$, and thus unbounded.
\end{theorem}

When it comes to the proportional-to-square-roots scheme, the picture is quite different. We can show that the Price of Stability is comparable to the Shapley scheme.

\begin{theorem}
\label{thm:PoS_square_roots}
Consider instances with $n$ players ($n=n_s+n_b$), so that $n_s \geq h^2 + n_b(\lceil h-\sqrt{a}+1 \rceil)$. Then, for the proportional-to-square-roots reward sharing scheme, we have that PoS = 2.
\end{theorem}

However, the proportional-to-square-roots scheme is not Sybil-proof w.r.t. its equilibria, in contrast to the Shapley value.

\begin{theorem}
\label{thm:sq_Sybil}
There exist instances where the proportional-to-square-roots scheme is not Sybil-proof w.r.t. $(k, l)$-equilibria, as given in Defintion \ref{def:k-l}. 
\end{theorem}

Hence we can conclude that for the family of all considered variants of proportional sharing, we either have vulnerability to Sybil attacks or high Price of Stability.

\section{Conclusions}
\label{sec:conclusions}

We have studied a model of pool formation over blockchains with large, influential players and players with negligible stake. The conclusions regarding the schemes we considered are as follows.

\begin{enumerate}[noitemsep]
    \item The Shapley scheme has constant Price of Stability bounded by 4/3 in the oceanic model and by $2$ in the atomic model. It is also Sybil-proof w.r.t. the constructed equilibria. Furthermore it is not always easy to find the best possible Sybil strategy.
    \item The proportional scheme always has an optimal equilibrium (PoS = 1). But it is not Sybil-proof w.r.t. its equilibria. Moreover, Sybil attacks can be computed relatively easily.
    \item The proportional to squares scheme has unbounded Price of Stability.
    \item The proportional to square roots scheme has constant Price of Stability bounded by 2. But it is not Sybil-proof w.r.t. its equilibria.
\end{enumerate}

We view this work as a promising initial step towards demonstrating the applicability of cooperative game theory concepts, such as the Shapley scheme, in the context of reward sharing for blockchain protocols. While not exhaustive, our findings provide a foundation for further exploration and development in this area.

\subsection{Future work}
\label{subsec:future}
    
There are several interesting avenues for future research. An immediate direction is to consider a different function $\rho(S)$ for the pool rewards so as to capture other scenarios or protocols. 
Our focus for $\rho(\cdot)$ in this paper was a
threshold function inspired of validator rewards
in Ethereum. However, other schemes are possible for instance, 
 we can consider a model where $\rho(S)$ is a linear (up to a threshold) function of $\sum_{i\in S} a_i$, or other variants that provide additional reward to pools with large pool operators so that some ``skin  in the game'' is present in pool formation, cf. \cite{brunjes2020reward}.  With respect to this latter property, it would be also interesting to identify skin in the game as a separate or alternative objective (i.e., separate from welfare and the 
 number of independent pools involved at equilibrium). For instance, having skin in the game  is reflected in design choices of pool formation contracts like Rocketpool to impose a minimum threshold of 8 ETH to the operator for each pool spun by the system. It is easy to capture such considerations in the oceanic model by observing that skin in the game is not provided by non-atomic players and hence a pool configuration can only derive skin in the game from the atomic players. Comparing  equilibria and the Price of Stability under this objective would be an interesting parallel to the results presented herein, while considering both objectives simultaneously (equilibria exhibiting  high decentralization and high skin in the game at the same time) points to an interesting Pareto optimization question. 
 
  Another direction is to consider further concepts from game theory, such as the Banzhaf index \cite{CW16}. This has played a prominent role in weighted voting games, and it would be interesting to explore if it would lead to a performance comparable to the Shapley scheme.
Finally, it is also important to investigate the dynamics for the iterated version of pool formation games and the speed of convergence to good equilibria.

\bibliographystyle{ACM-Reference-Format}
\input{Arxiv-version.bbl}

\appendix

\section{Discussion related to Remark \ref{rem:stake}}
\label{sec-app:big-players}

In this section, we revisit the assumption we have made that $a_i< h$. If we have very large players with $a_i > h$, obviously one choice for them is that they could split their stake into portions of size $h$ and run their own individual pools, which will not get involved in the pool formation game. Hence they would participate in pools with other players only with their remainder, which is $a_i ~mod~ h$. Given our equilibrium analysis, we provide evidence below that this is a meaningful action for them, in the sense that splitting their stake in portions smaller than $h$  (say $h/m$ for some $m>1$) and participate with multiple avatars into our pool formation game does not yield better outcomes for them. 

\begin{theorem}
    Consider a stake distribution where there exists a single player $i\in N_a$ with $a_i = \lambda h$ and for every $j\neq i$, $a_i < h$. Suppose that player $i$ is considering splitting her stake into portions of size $h/m$ and forming pools with other players, where $m>1$. Then, for the Shapley scheme, and the equilibria that arise under this splitting, as identified in Section \ref{sec:shapley}, player $i$ is not better off, compared to the scenario where she splits her stake into $\lambda$ portions of size exactly $h$ each and runs her own individual pools. The same also holds for any equilibrium under the proportional scheme. 
\end{theorem}

\begin{proof}
    If $i$ splits her stake into portions of size $h$, and runs her own pools, her payoff will be exactly equal to $\lambda$. We compare this against a uniform split into $\lambda m$ portions of size $h/m$ each, where $m>1$. Hence player $i$ participates with $\lambda m$ identities in the game.

    Consider the Shapley scheme. For any equilibrium in the form described by the conditions of Lemma \ref{lem:eq-conditions-oceanic}, each such avatar will form a pool with small players. By Lemma \ref{lemma:nonatomic-rewards}, the reward of each avatar will be equal to its Shapley value in a pool consisting of herself and a volume of mass $k_i$ for some $k_i\leq h$, which is equal to $\frac{k_i - h + h/m}{k_i}$. Hence, the total reward of player $i$, after summing up the reward of all her avatars will be equal to:
    $$\lambda m \cdot  \frac{k_i - h + h/m}{k_i}$$

    If we now check for what values of $m$ this total reward is no more than $\lambda$, this is equivalent to:
    $$ m \cdot  \frac{k_i - h + h/m}{k_i} \leq 1$$ 
    By simplifying the above expression, we get that this is equivalent to $k_i\leq h$, which is always true for these equilibria by Lemma \ref{lem:eq-conditions-oceanic}. Hence, $i$ cannot receive a higher reward by splitting her stake uniformly into portions of size strictly less than $h$.

    Finally, regarding the proportional scheme, it is very easy to verify that a player cannot become better off by splitting her stake into portions of size $h/m$ each, for $m>1$.
\end{proof}

\section{Missing proofs from Section \ref{sec:shapley}}
\label{app-sec:shapley}

\subsection{Proof of Lemma \ref{lem:typeA}}

Consider a formation where agent 1 forms a pool with $k< h$ other players. We analyze the Shapley formula of \eqref{eq:shapley} for the pool under consideration. The resulting payoffs occur from the fact that agent 1 will have marginal contribution equal to 1 in all the permutations where she is in a position $i$ with  $h-a<i\leq k+1$, which yields that the total number of these permutations is $k!\cdot (k-h+a+1)$. It suffices now to divide this by $(k+1)!$, which is the total number of all permutations of the $k+1$ players. 
The other $k$ players simply get the remaining reward divided by $k$, that is for $i\in S\setminus\{1\}$, we have $\phi_i(S)=\frac{1-\phi_1(S)}{k}=\frac{h-a}{k\cdot(k+1)}$. 

For the second case of the lemma, consider a formation where agent 1 forms a pool with $k\geq h$ other players. Agent 1 will have marginal contribution equal to 1, in all the permutations that she is placed in position i with $h-a\leq i\leq h$, hence the total number of these permutations is $a\cdot(k!)$, and dividing by the total number of all permutations $(k+1)!$, we get $\frac{a}{k+1}$. The other $k$ players, again receive the remaining reward divided by $k$, which is:$\frac{1-\frac{a}{k+1}}{k}=\frac{k-a+1}{k(k+1)}$.

\subsection{Proof of Lemma \ref{lem: 4.4}}

    Consider 2 such pools $S_1, S_2$, and let $|S_1| = t, |S_2|=r$. Suppose that $t\geq r+2$. Then, by Fact \ref{fact:small_pool}, for each $i\in S_1$, $\phi_i(S_1) = 1/t$. By deviating to join pool $S_2$, such a player will receive a payoff of $\phi_i(S_2\cup \{i\}) = 1/(r+1) > 1/t$. Hence this partition was not an equilibrium.   

    \subsection{Proof of Lemma \ref{lem:NE_conditions}}

The inequalities of the lemma represent the fact that no player should have an incentive to move from the pool that she belongs to. More specifically, the lower bound for $l$ in the first condition arises from the fact that none of the $k$ small players in the Type A pool should want to move to a pool of Type C, with $l-1$ players. By Lemma \ref{lem:typeA} and Fact \ref{fact:small_pool}, this is the case if $\frac{h-a}{k(k+1)} \geq 1/l$, which yields the desired lower bound for $l$. Note that this implies that these players also have no incentive to move to a Type B pool.

Consider now a deviation from a player of a Type B pool, going to the Type A pool. For such a deviation not to be profitable, we need to have that $1/l \geq \frac{h-a}{(k+1)(k+2)}$. This yields the upper bound of the first condition for $l$, and note also that if this holds, then players from a Type C pool also do not have an incentive to go to the Type A pool. 

So far we have shown that none of the small players, regardless of the pool they belong to, would have an incentive to deviate if the first condition of the lemma is satisfied. 
Lastly, the second condition of the lemma reflects the fact that player 1 should not desire to move to a Type C pool. In such a deviation, since $l-1\geq h$, by Lemma \ref{lem:typeA}, player 1 would receive $\frac{a}{l}$. Hence we need that  $\frac{k-h+a+1}{k+1} \geq \frac{a}{l}$, which yields precisely the desired condition.
Note also that if this holds, player 1 does not have an incentive to deviate to a Type B pool either.

Hence, if all the conditions of the lemma are true, there is no player who would have a profitable deviation.

 \subsection{Proof of Lemma \ref{lem:k,l values-part1}}   
   
 We start first with the following property, showing that indeed there are feasible values for $k$ that are at most $h-2$.
   \begin{lemma}
     \label{lem:k}
         It holds that $k^* =  \frac{1}{2}(\sqrt{-3a^2+h^2+2ah-2h+2a+1}-a+h-1)\leq h-2$, for any $a+1\leq h\leq a^2-2a+2$.
     \end{lemma}  
     \begin{proof}    
     Starting from what we want to prove: $\frac{1}{2}(\sqrt{-3a^2+h^2+2ah-2h+2a+1}-a+h-3)\leq h-2\Leftrightarrow \sqrt{-3a^2+h^2+2ah-2h+2a+1}\leq h+a-1\Leftrightarrow -3a^2+h^2+2ah-2h+2a+1\leq (h+a-3)^2$, which holds for any $a\geq 3$ and $a+1\leq h\leq a^2-2a+2$. \qed
    \end{proof}
    
Since $k^*\leq h-2$, we also have that $\lceil k^* \rceil \leq h-2$ and combining this with the fact that $a\leq h-1$, we get that $\lceil k^* \rceil +a<2h$. Thus, with $k=\lceil k^* \rceil$, the Type A pool that we construct satisfies the properties we need.
    
We can now substitute $k$ with $\lceil k^*\rceil$ in the first condition of Lemma \ref{lem:NE_conditions}, and we can get the lowest feasible integer value of $l$, for which we have an equilibrium.
In particular, since the interval $\Delta$ contains two integers, we set
$$ l^* = \frac{(\lceil k^*\rceil)(\lceil k^*\rceil+1)}{h-a}, ~~~l = \lceil l^* \rceil + 1 $$
In accordance with the definition of a $(k, l)$-partition, we now need to show that $l-1 = \lceil l^*\rceil \geq h$ and at the same time, for the Price of Stability guarantee we need that  $\lceil l^*\rceil + 1 \leq 2h$. 
Since $k^* \leq \lceil k^* \rceil <k^*+1$, it suffices to show that $h\leq \frac{k^*(k^*+1)}{h-a}$ and that 
$\frac{(k^*+1)(k^*+2)}{h-a}\leq 2h$.\\
  When $a=h-1$, this is easy to establish as Equation \eqref{k*_value} simplifies to $k\geq \sqrt{h-1}$. So we can set $k= \lceil \sqrt{h-1} \rceil$. Hence,  $\frac{k(k+1)}{h-a}=k(k+1)\geq(\sqrt{h-1})(\sqrt{h-1}+1)\geq h$ for any $h\geq2$. Moreover, $\frac{k(k+1)}{h-a}\leq (\sqrt{h-1}+1)(\sqrt{h-1}+2)=h+3\sqrt{h-1}+1$ which is lower or equal to $2h$ for any $h\geq 10$, also for $2\leq h \leq  10$, one can easily check that $(\lceil\sqrt{h-1}\rceil)(\lceil\sqrt{h-1}\rceil+1)\leq 2h$. 
  
  When $a< h-1$, the properties we need are verified by the following claim. 
    \begin{claim}
     For $k=k^*$, it holds that: $h\leq\frac{k^*(k^*+1)}{h-a} \leq \frac{k(k+1)}{h-a}\leq\frac{(k^*+1)(k^*+2)}{h-a}\leq 2h$,  for any $3\leq a\leq h-2$ and $h\leq a^2-2a+2$.
    \end{claim}
    
\subsection{Proof of Lemma \ref{lem:k,l values-part2}}

For this case we will use $k=h-1$. From Lemma \ref{lem:NE_conditions} we get that when $k=h-1$ it must hold that $\frac{a}{h}\geq \frac{a}{l}$ and $\frac{h(h-1)}{h-a}\leq l \leq \frac{h(h+1)}{h-a+1}.$
Therefore, we have an interval $\Delta$ for $l$, and we need to identify an integer value in this interval.

First, we observe that when $h\geq a^2$, a valid integer value for $l$ is the value $h+a$. 
This at least  equal to $h+1$ for all $a\geq 2$ and hence the first condition is met. Moreover, it holds that $\frac{h(h-1)}{h-a}\leq h+a$, for any $h\geq a^2$ and also $h+a\leq \frac{h(h+1)}{h-a+1}$ for any $l$ and $a\geq 1$. This the value $h+a$ belongs to the interval $\Delta$ for $h\geq a^2$. Finally, $a+h\leq 2h$ for $a\leq h-1$. Hence, the tuple $k=h-1, l=h+a$ is an appropriate choice.

Next we also check the  range of $h$ for which the value $l=h+a+1$ is a valid solution. This value is at least $h$ and also lower or equal to $2h$ for all $2\leq a\leq h-1$. Additionally, $\frac{h(h-1)}{h-a}\leq h+a+1$ for any $h\geq \frac{a^2+a}{2} $ and also $h+a+1\leq \frac{h(h+1)}{h-a+1}$ for any $h\leq a^2-1$. Therefore, the tuple $k=h-1, l=h+a+1$ is an appropriate choice as long as $h\in [\frac{a^2+a}{2}, a^2-1]$.

Combining the above, we have showed the existence of integers values $k,l$, that satisfy the conditions of Lemma \ref{lem:NE_conditions}, for the cases where  $\frac{a^2+a}{2}\leq h\leq a^2-1$ and $h\geq a^2$. The region of values for $h$, that remains is  $a^2-2a+2< h<\frac{a^2+a}{2}$.
But now we can notice that $\frac{a^2+a}{2}\leq a^2-2a+2$ for any $a\geq 4$. Also $a=2$ and $a=3$ there are no integer values for $h$ included in this region. Hence, there is nothing else to check. and we have covered the entire range for $h$ and $a$ that concern Case 2.

\subsection{Proof of Theorem \ref{thm:pos-general}}
We will consider the construction of a $(k, l)$-partition as an equilibrium with $k\leq h-1$, and $k>k^*$, where $k^*$ is as defined in the proof of Theorem \ref{thm:PoS}. In order to argue about deviations of a large player from her pool to another Type A pool, we need to compute the Shapley value for pools where there are 2 large players together with $k$ small players. 
Consider a large player $i$ belonging to pool $S_i$, that we will refer to as the deviator, who is considering to move to a different Type A pool, say $S$. We show in the next lemma, that under the conditions of the theorem, this is not a profitable deviation.

\begin{lemma}
    For $h-1\geq k\geq k^*=\frac{1}{2}(\sqrt{-3a^2+2ah+2a+h^2-2h+1}-a+h-1)$, it holds that $\phi_i(S \cup \{i\}) \leq \phi_i(S_i)$, for any $h,a$ with $2\leq a<h$.
\end{lemma}
\begin{proof}
    
We will distinguish two cases, the case of $a\geq \frac{h}{2}$ and that of $a<\frac{h}{2}$. 

For the first case, $a\geq \frac{h}{2}$, we claim that the reward player $i$ would get, if she deviates to $S$, is:
\begin{align}
\label{Shapley_a>=h/2}
    \phi_i(S \cup \{i\})=\frac{(h-a)(h-a+1)k!+(k-h+a+1)(k-h+a+2)k!}{2((k+2)!)}.
\end{align}
The first term of the numerator, reflects the total number of permutations, where the deviator is positioned after the other large player with stake $a$ and has a marginal contribution equal to 1. These are the permutations, where the sum of the stakes of agents positioned before the deviator, is greater or equal to $a$ and lower or equal to $h-1$. Similarly, the second term represents the total number of permutations, where the deviator is positioned before the other large player and has a marginal contribution equal to 1. In these permutations the player with stake $a$ is positioned after the $(h-a+2)$-th position. This formula can be simplified to $\phi_i(S \cup \{i\})=\frac{(h-a)(h-a+1)+(k-h+a+1)(k-h+a+2)}{2(k+1)(k+2)}$.

By Lemma \ref{lem:typeA}, we also know that $\phi_i(S_i) = \frac{k-h+a+1}{k+1}$. Therefore, what we need is to show that the following inequality holds:

\begin{align*}
      \frac{(h-a)(h-a+1)+(k-h+a+1)(k-h+a+2)}{2(k+1)(k+2)}\leq \frac{k-h+a+1}{k+1},  
    \end{align*}

    To see this, the above inequality is equivalent to $(h-a)(h-a+1)+(k-h+a+1)(k-h+a+2)\leq 2(k+2)(k-h+a+1)$. In order to prove that this holds, it suffices to show the validity of the following two simpler inequalities, and then add them up. Namely, the first one is 
    \begin{align*}
         (k+2)(k-h+a+1)\geq (k-h+a+1)(k-h+a+2)\Leftrightarrow k+2\geq k-h+a+2
    \end{align*} 
    and the second one is 
    \begin{align*}
      (k+2)(k-h+a+1)\geq (h-a)(h-a+1)  
    \end{align*}
    For the first one, it is easy to see that it holds due to the fact that $h\geq a\Rightarrow k+2\geq k-h+a+2$. For the other inequality, for the values of $k$ that we are interested in, we know from Section \ref{subsec:shapley-1}, that for $k\geq k^*$, it holds that $\frac{k(k+1)}{h-a}\geq \frac{a(k+1)}{k-h+a+1}\Rightarrow k-h+a+1\geq \frac{a(h-a)}{k}$, so it suffices to show that $\frac{a(h-a)}{k}\geq \frac{(h-a)(h-a+1)}{k+2}$. This last inequality, is equivalent to $a(k+2)\geq k(h-a+1)\Leftrightarrow ak+2a\geq kh-ak+k\Leftrightarrow 2ak+2a\geq kh+k$, which holds due to the fact that $a\geq \frac{h}{2}\Rightarrow 2ak\geq kh$ and also since $k\leq h\Rightarrow k\leq 2a$. This completes the proof for the case of $a\geq \frac{h}{2}$.

For the second case, where $a<\frac{h}{2}$, the reward the deviator would get by moving to pool $S$ is:
\begin{align}
\label{Shapley_a<h/2}
    \phi_i(S \cup \{i\})=\frac{(a)(2h-3a+1)k!+(k-h+a+1)(k-h+a+2)k!}{2((k+2)!)}.
\end{align}
Again, the first term of the numerator, represents the total number of permutations,where the deviator is positioned after the other player with stake $a$ and has a marginal contribution equal to 1. Also, the second term represents the total number of permutations, where the deviator is positioned before the other large player and has a marginal contribution equal to 1. Hence, what we need to show in this case boils down to the following inequality.
\begin{align*}    
    \frac{(2h-3a+1)a+(k-h+a+1)(k-h+a+2)}{2(k+1)(k+2)}\leq \frac{k-h+a+1}{k+1},
    \end{align*}
This can be verified in a similar way as the analogous inequality in the first case of $a\geq h/2$.
\end{proof}

In addition to the above lemma, we also need to show that none of the $k$ small players from a Type A pool have an incentive to move to another Type A pool. But this easy to see, as by doing so, their reward would become $\frac{h-a}{k+2}$, which is strictly less than $\frac{h-a}{k+1}$. 

Hence, for the construction that we had in the proof of Theorem \ref{thm:PoS}, where there was only one Type A pool, we have now shown that even when there are more Type A pools, one per each large player, we are still at an equilibrium. Therefore, by adjusting the lower bound on the number of small players that we need, so that such a construction is feasible, we retain the same guarantee for the Price of Stability, and the proof of the theorem is complete.

 \section{Missing proofs from Section \ref{sec:Sybil}}

\subsection{Proof of Theorem \ref{thm:water-filling}}
\label{proof-water-filling}

Let  $C_1,\dots, C_m$ be the pools formed by the other players, and consider an agent $i$ with stake $a_i$.  First, as already mentioned before the statement of Theorem \ref{thm:water-filling}, the objective is a sum of concave functions, and thus concave, therefore we could solve it up to any desired accuracy (since this is equivalent to minimizing a convex function subject to linear constraints). 

A different and simpler way to solve this is to use a greedy algorithm. 
The objective function is $\sum_{j=1}^m \frac{s_{i_j}}{m(C_j) + s_{i_j}}$. This is an {\it additively separable} function, i.e, a sum of functions, each of which is concave w.r.t. a different variable $s_{i_j}$.
Therefore, one should start with allocating stake to the pool that corresponds to the term with the highest derivative. But this is precisely the pool with the lowest stake. Hence we can allocate stake to this pool until it equalizes the second lowest pool. From that point onwards, the player has to allocate her stake equally to these two pools, which can be thought of as water filling these two pools at the same rate. And this process has to continue in the same flavor, involving more and more pools, until all of $a_i$ is allocated.

To be more precise, we simply need to run a for loop at most $m$ times, and in each iteration, we commit a portion of stake to the pool (or pools) with the lowest total stake, since this offers the greater rate of reward per unit. 
At the first iteration let $C_p$ be the pool with the lowest total stake and $C_k$ be the pool with the second lowest total stake. Then player $i$ commits stake equal to $m(C_k) - m(C_p)$, as long as $a_i$ is greater or equal to this quantity, otherwise, she commits all her remaining stake to $C_p$ and the procedure terminates. 
Moving on, at any other iteration $t$, there will be more than one pool with the lowest total stake, namely $C_{p_1}, C_{p_2},\dots, C_{p_t}$. Then agent $i$ splits her remaining stake equally to all of them, until all of them reach a total stake equal to the pool with the second lowest stake, and then move to the next iteration. If the remaining stake $s_i'$ does not suffice to do so, she splits $s_i'$ equally among pools $C_{p_1}, C_{p_2},\dots, C_{p_t}$ and the procedure terminates.

\noindent The above procedure is obviously polynomial with respect to the number of pools $m$.

\section{Missing proofs from Section \ref{sec:alt}}
\subsection{Proof of Theorem \ref{thm:sq_lb}}
    
Consider an instance where there is only one large player with stake $a=h-1$ and $n-1$ small players with stake 1. In a formation that is an equilibrium, at least one small player needs to be in the same pool with the agent that has stake $a$. That small player will get reward lower or equal to $\frac{1}{(h-1)^2+1}$, hence in order to not have an incentive to move to another pool (consisting of only small players), all such pools need to have at least $(h-1)^2$ agents. This allows us to construct an equilibrium  for instances with $n=(h-1)^2+2$. We will have the large player form one pool together with one small player. All the other small players form one additional pool of stake $(h-1)^2$. By the preceding dicsussion, this is an equilibrium with only two pools. On the other hand, the optimal formation will have $\lfloor \frac{a + 1 + (h-1)^2}{h}\rfloor = \lfloor \frac{h + (h-1)^2}{h}\rfloor \geq h-1$. 
Therefore, the Price of Stability for the proportional-to-squares reward sharing scheme would be at least $\frac{h-1}{2} = \Omega(h)$.

\subsection{Proof of Theorem \ref{thm:PoS_square_roots}}
\label{appendix-square-roots}
 
We will show the existence of $(k, l)$-equilibria, for the proportional-to-square-roots scheme, as we did with the Shapley scheme in Theorem \ref{thm:pos-general}, but for a different range of values for $k$ and $l$. 
 The next lemma identifies sufficient conditions for the existence of a $(k, l)$-equilibrium, which are derived by the equations describing that no player has an incentive to deviate.
\begin{lemma}
\label{lem:NE_conditions_sqrt}[Sufficient conditions]
    A $(k, l)$-partition is a Nash equilibrium under the proportional-to-square-roots scheme, if the following conditions hold:
    \begin{enumerate}[noitemsep]
        \item $\frac{\sqrt{a}}{\sqrt{a}+k}\geq \frac{\sqrt{a}}{\sqrt{a}+l-1}\Leftrightarrow l-1\geq k $
        \item $\frac{1}{\sqrt{a}+k}\geq\frac{1}{l}\Leftrightarrow k\leq l-\sqrt{a}$
       \item $\frac{1}{l}\geq \frac{1}{\sqrt{a}+k+1}\Leftrightarrow k\geq l-\sqrt{a}-1$
        
    \end{enumerate}
\end{lemma}
\begin{proof}
First, note that no player has an incentive to abandon a pool and start a pool on her own, since $a< h$.
The first condition, represents the fact that a big player with stake $a$ should not have an incentive to move from a Type A pool to a Type C pool (the latter implies that she also has no incentive to move to a type B pool).  The second condition states that no small player in a type A pool has an incentive to move to a type C pool (and this again implies that she has no incentive to move to a type B pool). The last condition represents the absence of motive for a small player to move from a type B pool to a type A pool. This also implies the same for the small players who belong to a type C pool.   The condition, which represents the fact that a big player with stake $a$ does not have an incentive to move to a Type A pool of another big player, is omitted since $\frac{\sqrt{a}}{\sqrt{a}+k}\geq \frac{\sqrt{a}}{\sqrt{a}+\sqrt{a}+k}\Leftrightarrow \sqrt{a}\geq 0$, always holds.  
\end{proof}

Using the previous lemma, the next step is to derive a general upper bound on the Price of Stability as a function of the relevant parameters of the game.

\begin{theorem}
\label{PoS_square_roots}
    Consider instances with $n$ players ($n=n_s+n_b$), so that $n_s \geq h^2 + n_b(\lceil h-\sqrt{a}+1 \rceil)$. Then, the proportional-to-square-roots reward sharing scheme has Price of Stability bounded by 
    $$ PoS \leq \frac{h+1}{h}\cdot\frac{(n_b\cdot a+n_s)}{(n_b\cdot(\sqrt{a}-1)+n_s-h-1)}.$$
\end{theorem}

\begin{proof}

We will show how to set values for $k$ and $l$ so that the conditions of Lemma \ref{lem:NE_conditions_sqrt} are satisfied. From the last two conditions of Lemma \ref{lem:NE_conditions_sqrt}, we get that for the value of $k$ it must hold that $l-\sqrt{a}\leq k\leq l-\sqrt{a}+1$. Hence, there is room for $k$ to have an integer value. For $l=h+1$, which is the lowest possible value for $l$ (since it must hold that $l-1\geq h$), we get that the lowest integer value for $k$ is $\lceil h-\sqrt{a}+1 \rceil$. We note that for this value of $k$ it can be verified that the first condition of Lemma \ref{lem:NE_conditions_sqrt} is also satisfied, and also that $k+a\geq h$, which is a requirement of the lemma. 
Therefore, given the values of $l=h+1$ and $k= \lceil h-\sqrt{a}+1 \rceil $, and if the number of small players is sufficiently large, we can construct a $(k, l)$-partition in the same way as in the last part in the proof of Theorem \ref{thm:PoS}. 
Recalling that  $n_b$ is the number of the big players, the total number of pools created for the values of $k$ and $l$ we use is at least:  $n_b+ \lfloor \frac{n-n_b-n_b\cdot(h-\sqrt{a}+2)}{h+1} \rfloor = \lfloor \frac{n_s+n_b\cdot(\sqrt{a}-1)}{h+1}\rfloor$. 

\noindent Coming now to the optimal formation, we know it always satisfies that $OPT(G) \leq \lfloor (\sum s_i)/h \rfloor$. Hence this means that  $OPT(G) \leq \lfloor\frac{n_b\cdot a+n_s}{h}\rfloor$. This yields that the Price of Stability for the  proportional-to-square-roots reward sharing scheme is at most $\frac{\frac{n_b\cdot a+n_s}{h}}{\frac{n_s+n_b\cdot \sqrt{a}}{h+1}-1}= \frac{h+1}{h}\cdot\frac{(n_b\cdot a+n_s)}{(n_b\cdot(\sqrt{a}-1)+n_s-h-1)}$.
\end{proof}

\noindent From the previous theorem, we can actually extract a constant upper bound of 2 on PoS. Under the conditions of Theorem \ref{PoS_square_roots}, we show that there always exists an equilibrium with $l-2=h$ and $k=\lceil h-\sqrt{a}+1\rceil\leq h-\sqrt{a}+2\leq h+2 $. 
To establish that we have an equilibrium for these choices of $k$ and $l$, it suffices to show that for $l=h+2$ and $k=\lceil h-\sqrt{a}+1\rceil$, the conditions of Lemma \ref{lem:NE_conditions_sqrt} are met. Firstly, for $l=h+2$ the first condition is met, since $h+2>k+1$. Moreover, for the value of $k$ we have that $\lceil h-\sqrt{a}+1\rceil\geq h-\sqrt{a}+1=l-\sqrt{a}-$ and also $\lceil h-\sqrt{a}+1\rceil\leq h-\sqrt{a}+2= l-\sqrt{a}$. Hence, the last two conditions are also satisfied. 

Finally, it can be verified that all pools in this $(k, l)$-partition have total stake at most $2h$.
This is because $l=h+2\leq 2h$ for any $h\geq 2$ and also $k+a=\lceil h-\sqrt{a}+1\rceil+a\leq 2h$, for any $2\leq a<h$. Therefore, this implies that the Price of Stability is at most 2.

At the same time, we can also have a matching lower bound on PoS, making the previous upper bound tight.
A simple instance in which the scheme achieves Price of Stability equal to 2 is the following: there is one player with stake $a$ and $2h-a$ players with stake 1. The optimal formation consists of 2 pools with total stake exactly $h$. From Lemma \ref{lem:NE_conditions_sqrt}, in order for a formation with $l=h$ to be an equilibrium, it must hold that $k\geq l-\sqrt{a}-1=h-\sqrt{a}-1\ge h-a$, for any $a\geq 3$ and hence there are not enough small players for a Nash equilibrium with 2 pools to exist. Thus only the grand coalition is an equilibrium, and hence PoS = 2.

 \noindent This concludes the proof of Theorem \ref{thm:PoS_square_roots}.

\subsection{Proof of Theorem \ref{thm:sq_Sybil}}

In order to show this, we use the $(k, l)$-equilibrium identified in Appendix \ref{appendix-square-roots} with $k=\lceil h-\sqrt{a}-1\rceil$. The reward of each large agent then is $\frac{\sqrt{a}}{ \sqrt{a}+h-\lceil\sqrt{a}\rceil-1}$. Consider the scenario where a large agent $i$ splits her stake into single unit amounts and commits them to other Type A pools (we consider an instance where there are at least $a+1$ large players). The total reward she receives in that case is: $\frac{a}{\sqrt{a}+ h-\lceil\sqrt{a}\rceil }$.
We can now check that 
$\frac{a}{ \sqrt{a}+h-\lceil\sqrt{a}\rceil } > \frac{\sqrt{a}}{ \sqrt{a}+h-\lceil\sqrt{a}\rceil-1}$, for any $2\leq a<h$, thus the player receives more than committing to a single pool.

\end{document}

%% file: Arxiv-version.bbl

%% file: Arxiv-version.bbl
\begin{thebibliography}{28}


\ifx \showCODEN    \undefined \def \showCODEN     #1{\unskip}     \fi
\ifx \showDOI      \undefined \def \showDOI       #1{#1}\fi
\ifx \showISBNx    \undefined \def \showISBNx     #1{\unskip}     \fi
\ifx \showISBNxiii \undefined \def \showISBNxiii  #1{\unskip}     \fi
\ifx \showISSN     \undefined \def \showISSN      #1{\unskip}     \fi
\ifx \showLCCN     \undefined \def \showLCCN      #1{\unskip}     \fi
\ifx \shownote     \undefined \def \shownote      #1{#1}          \fi
\ifx \showarticletitle \undefined \def \showarticletitle #1{#1}   \fi
\ifx \showURL      \undefined \def \showURL       {\relax}        \fi
\providecommand\bibfield[2]{#2}
\providecommand\bibinfo[2]{#2}
\providecommand\natexlab[1]{#1}
\providecommand\showeprint[2][]{arXiv:#2}

\bibitem[Arnosti and Weinberg(2022)]%
        {AW22}
\bibfield{author}{\bibinfo{person}{Nick Arnosti} {and}
  \bibinfo{person}{S.~Matthew Weinberg}.} \bibinfo{year}{2022}\natexlab{}.
\newblock \showarticletitle{Bitcoin: {A} Natural Oligopoly}.
\newblock \bibinfo{journal}{\emph{Management Science}} \bibinfo{volume}{68},
  \bibinfo{number}{7} (\bibinfo{year}{2022}), \bibinfo{pages}{4755--4771}.
\newblock


\bibitem[Assmann et~al\mbox{.}(1984)]%
        {AJKL84}
\bibfield{author}{\bibinfo{person}{S.~F. Assmann}, \bibinfo{person}{D.~S.
  Johnson}, \bibinfo{person}{D.~J. Kleitman}, {and} \bibinfo{person}{J.~Y.
  Leung}.} \bibinfo{year}{1984}\natexlab{}.
\newblock \showarticletitle{On a Dual Version of the One-Dimensional Bin
  Packing Problem}.
\newblock \bibinfo{journal}{\emph{Journal of Algorithms}}  \bibinfo{volume}{5}
  (\bibinfo{year}{1984}), \bibinfo{pages}{502--525}.
\newblock


\bibitem[Ausubel and Milgrom(2006)]%
        {AM01}
\bibfield{author}{\bibinfo{person}{L.~M. Ausubel} {and} \bibinfo{person}{P.
  Milgrom}.} \bibinfo{year}{2006}\natexlab{}.
\newblock \showarticletitle{The Lovely but Lonely Vickrey Auction}.
\newblock In \bibinfo{booktitle}{\emph{Combinatorial Auctions}},
  \bibfield{editor}{\bibinfo{person}{Peter Cramton}, \bibinfo{person}{Yoav
  Shoham}, {and} \bibinfo{person}{Richard Steinberg}} (Eds.).
  \bibinfo{publisher}{MIT Press}, \bibinfo{pages}{18--40}.
\newblock


\bibitem[Azouvi and Hicks(2022)]%
        {AH22}
\bibfield{author}{\bibinfo{person}{Sarah Azouvi} {and}
  \bibinfo{person}{Alexander Hicks}.} \bibinfo{year}{2022}\natexlab{}.
\newblock \showarticletitle{Decentralisation Conscious Players and System
  Reliability}. In \bibinfo{booktitle}{\emph{26th International Conference on
  Financial Cryptography and Data Security, {FC} 2022}}.
  \bibinfo{publisher}{Springer}, \bibinfo{pages}{426--443}.
\newblock


\bibitem[Bachrach et~al\mbox{.}(2010)]%
        {BMRPRS10}
\bibfield{author}{\bibinfo{person}{Yoram Bachrach}, \bibinfo{person}{Evangelos
  Markakis}, \bibinfo{person}{Ezra Resnick}, \bibinfo{person}{Ariel~D.
  Procaccia}, \bibinfo{person}{Jeffrey~S. Rosenschein}, {and}
  \bibinfo{person}{Amin Saberi}.} \bibinfo{year}{2010}\natexlab{}.
\newblock \showarticletitle{Approximating power indices: theoretical and
  empirical analysis}.
\newblock \bibinfo{journal}{\emph{Auton. Agents Multi Agent Syst.}}
  \bibinfo{volume}{20}, \bibinfo{number}{2} (\bibinfo{year}{2010}),
  \bibinfo{pages}{105--122}.
\newblock


\bibitem[Bahrani et~al\mbox{.}(2024)]%
        {BGR24}
\bibfield{author}{\bibinfo{person}{Maryam Bahrani}, \bibinfo{person}{Pranav
  Garimidi}, {and} \bibinfo{person}{Tim Roughgarden}.}
  \bibinfo{year}{2024}\natexlab{}.
\newblock \bibinfo{title}{Centralization in Block Building and Proposer-Builder
  Separation}.
\newblock
\newblock
\newblock
\shownote{arXiv:2401.12120}.


\bibitem[Br{\"u}njes et~al\mbox{.}(2020)]%
        {brunjes2020reward}
\bibfield{author}{\bibinfo{person}{Lars Br{\"u}njes}, \bibinfo{person}{Aggelos
  Kiayias}, \bibinfo{person}{Elias Koutsoupias}, {and}
  \bibinfo{person}{Aikaterini-Panagiota Stouka}.}
  \bibinfo{year}{2020}\natexlab{}.
\newblock \showarticletitle{Reward sharing schemes for stake pools}. In
  \bibinfo{booktitle}{\emph{2020 IEEE european symposium on security and
  privacy (EuroS\&p)}}. IEEE, \bibinfo{pages}{256--275}.
\newblock


\bibitem[Buterin(2013)]%
        {Ethereum}
\bibfield{author}{\bibinfo{person}{Vitalik Buterin}.}
  \bibinfo{year}{2013}\natexlab{}.
\newblock \bibinfo{title}{A Next-Generation Smart Contract and Decentralized
  Application Platform}.
\newblock
\newblock
\newblock
\shownote{\url{https://github.com/ethereum/wiki/wiki/White-Paper}}.


\bibitem[Buterin et~al\mbox{.}(2019)]%
        {casper-incentives}
\bibfield{author}{\bibinfo{person}{Vitalik Buterin},
  \bibinfo{person}{Dani{\"{e}}l Reijsbergen}, \bibinfo{person}{Stefanos
  Leonardos}, {and} \bibinfo{person}{Georgios Piliouras}.}
  \bibinfo{year}{2019}\natexlab{}.
\newblock \showarticletitle{Incentives in Ethereum's Hybrid Casper Protocol}.
  In \bibinfo{booktitle}{\emph{{IEEE} International Conference on Blockchain
  and Cryptocurrency, {ICBC} 2019, Seoul, Korea (South), May 14-17, 2019}}.
  \bibinfo{publisher}{{IEEE}}, \bibinfo{pages}{236--244}.
\newblock
\urldef\tempurl%
\url{https://doi.org/10.1109/BLOC.2019.8751241}
\showDOI{\tempurl}


\bibitem[Can et~al\mbox{.}(2022)]%
        {CHP22}
\bibfield{author}{\bibinfo{person}{Burak Can}, \bibinfo{person}{Jens~Leth
  Hougaard}, {and} \bibinfo{person}{Mohsen Pourpouneh}.}
  \bibinfo{year}{2022}\natexlab{}.
\newblock \showarticletitle{On reward sharing in blockchain mining pools}.
\newblock \bibinfo{journal}{\emph{Games Econ. Behav.}}  \bibinfo{volume}{136}
  (\bibinfo{year}{2022}), \bibinfo{pages}{274--298}.
\newblock


\bibitem[Chalkiadakis and Woolridge(2016)]%
        {CW16}
\bibfield{author}{\bibinfo{person}{G. Chalkiadakis} {and} \bibinfo{person}{M.
  Woolridge}.} \bibinfo{year}{2016}\natexlab{}.
\newblock \showarticletitle{Weighted Voting Games}.
\newblock In \bibinfo{booktitle}{\emph{Handbook of Computational Social
  Choice}}, \bibfield{editor}{\bibinfo{person}{Felix Brandt},
  \bibinfo{person}{Vincent Conitzer}, \bibinfo{person}{Ulle Endriss},
  \bibinfo{person}{J{\'{e}}r{\^{o}}me Lang}, {and} \bibinfo{person}{Ariel~D.
  Procaccia}} (Eds.). \bibinfo{publisher}{Cambridge University Press},
  \bibinfo{pages}{377--396}.
\newblock


\bibitem[Chen et~al\mbox{.}(2019)]%
        {chen2019axiomatic}
\bibfield{author}{\bibinfo{person}{Xi Chen}, \bibinfo{person}{Christos
  Papadimitriou}, {and} \bibinfo{person}{Tim Roughgarden}.}
  \bibinfo{year}{2019}\natexlab{}.
\newblock \showarticletitle{An axiomatic approach to block rewards}. In
  \bibinfo{booktitle}{\emph{Proceedings of the 1st ACM Conference on Advances
  in Financial Technologies}}. \bibinfo{pages}{124--131}.
\newblock


\bibitem[Chen et~al\mbox{.}(2020)]%
        {CSSZ20}
\bibfield{author}{\bibinfo{person}{Zhihuai Chen}, \bibinfo{person}{Xiaoming
  Sun}, \bibinfo{person}{Xiaohan Shan}, {and} \bibinfo{person}{Jialin Zhang}.}
  \bibinfo{year}{2020}\natexlab{}.
\newblock \showarticletitle{Decentralized Mining Pool Games in Blockchain}. In
  \bibinfo{booktitle}{\emph{IEEE International Conference on Knowledge Graph,
  {ICKG} 2020}}. \bibinfo{pages}{426--432}.
\newblock


\bibitem[Deng and Papadimitriou(1994)]%
        {DP94}
\bibfield{author}{\bibinfo{person}{Xiaotie Deng} {and}
  \bibinfo{person}{Christos Papadimitriou}.} \bibinfo{year}{1994}\natexlab{}.
\newblock \showarticletitle{On the complexity of cooperative solution
  concepts}.
\newblock \bibinfo{journal}{\emph{Mathematics of Operations Research}}
  \bibinfo{volume}{19}, \bibinfo{number}{2} (\bibinfo{year}{1994}),
  \bibinfo{pages}{257--266}.
\newblock


\bibitem[Douceur(2002)]%
        {douceur02}
\bibfield{author}{\bibinfo{person}{John~R. Douceur}.}
  \bibinfo{year}{2002}\natexlab{}.
\newblock \showarticletitle{The Sybil Attack}. In
  \bibinfo{booktitle}{\emph{Peer-to-Peer Systems, First International Workshop,
  {IPTPS} 2002}}. \bibinfo{pages}{251--260}.
\newblock


\bibitem[Fatima et~al\mbox{.}(2008)]%
        {fatima2008linear}
\bibfield{author}{\bibinfo{person}{Shaheen~S Fatima}, \bibinfo{person}{Michael
  Wooldridge}, {and} \bibinfo{person}{Nicholas~R Jennings}.}
  \bibinfo{year}{2008}\natexlab{}.
\newblock \showarticletitle{A linear approximation method for the {S}hapley
  value}.
\newblock \bibinfo{journal}{\emph{Artificial Intelligence}}
  \bibinfo{volume}{172}, \bibinfo{number}{14} (\bibinfo{year}{2008}),
  \bibinfo{pages}{1673--1699}.
\newblock


\bibitem[Gillies(1953)]%
        {gillies1953some}
\bibfield{author}{\bibinfo{person}{Donald~Bruce Gillies}.}
  \bibinfo{year}{1953}\natexlab{}.
\newblock \bibinfo{booktitle}{\emph{Some theorems on n-person games}}.
\newblock \bibinfo{publisher}{Princeton University}.
\newblock


\bibitem[Kiayias et~al\mbox{.}(2024)]%
        {marmolejo}
\bibfield{author}{\bibinfo{person}{Aggelos Kiayias}, \bibinfo{person}{Elias
  Koutsoupias}, \bibinfo{person}{Francisco Marmolejo-Cossio}, {and}
  \bibinfo{person}{Aikaterini-Panagiota Stouka}.}
  \bibinfo{year}{2024}\natexlab{}.
\newblock \showarticletitle{Balancing Participation and Decentralization in
  Proof-of-Stake Cryptocurrencies}. In \bibinfo{booktitle}{\emph{17th
  International Conference on Algorithmic Game Theory, {SAGT} 2024}}.
  \bibinfo{pages}{333--350}.
\newblock


\bibitem[Kwon et~al\mbox{.}(2019)]%
        {kwon2019impossibility}
\bibfield{author}{\bibinfo{person}{Yujin Kwon}, \bibinfo{person}{Jian Liu},
  \bibinfo{person}{Minjeong Kim}, \bibinfo{person}{Dawn Song}, {and}
  \bibinfo{person}{Yongdae Kim}.} \bibinfo{year}{2019}\natexlab{}.
\newblock \showarticletitle{Impossibility of full decentralization in
  permissionless blockchains}. In \bibinfo{booktitle}{\emph{Proceedings of the
  1st ACM Conference on Advances in Financial Technologies}}.
  \bibinfo{pages}{110--123}.
\newblock


\bibitem[Leonardos et~al\mbox{.}(2019)]%
        {LLP19}
\bibfield{author}{\bibinfo{person}{Nikos Leonardos}, \bibinfo{person}{Stefanos
  Leonardos}, {and} \bibinfo{person}{Georgios Piliouras}.}
  \bibinfo{year}{2019}\natexlab{}.
\newblock \showarticletitle{Oceanic Games: Centralization Risks and Incentives
  in Blockchain Mining}. In \bibinfo{booktitle}{\emph{1st International
  Conference on Mathematical Research for Blockchain Economy, {MARBLE} 2019}}
  \emph{(\bibinfo{series}{Springer Proceedings in Business and Economics})}.
  \bibinfo{publisher}{Springer}, \bibinfo{pages}{183--199}.
\newblock


\bibitem[Lewenberg et~al\mbox{.}(2015)]%
        {LB+15}
\bibfield{author}{\bibinfo{person}{Yoad Lewenberg}, \bibinfo{person}{Yoram
  Bachrach}, \bibinfo{person}{Yonatan Sompolinsky}, \bibinfo{person}{Aviv
  Zohar}, {and} \bibinfo{person}{Jeffrey~S. Rosenschein}.}
  \bibinfo{year}{2015}\natexlab{}.
\newblock \showarticletitle{Bitcoin Mining Pools: {A} Cooperative Game
  Theoretic Analysis}. In \bibinfo{booktitle}{\emph{Proceedings of the 14th
  International Conference on Autonomous Agents and Multiagent Systems, {AAMAS}
  2015}}. \bibinfo{publisher}{{ACM}}, \bibinfo{pages}{919--927}.
\newblock


\bibitem[Milnor and Shapley(1978)]%
        {MS78}
\bibfield{author}{\bibinfo{person}{J.~W. Milnor} {and} \bibinfo{person}{L.~S.
  Shapley}.} \bibinfo{year}{1978}\natexlab{}.
\newblock \showarticletitle{Values of Large Games II: Oceanic Games}.
\newblock \bibinfo{journal}{\emph{Mathematics of Operations Research}}
  \bibinfo{volume}{3}, \bibinfo{number}{4} (\bibinfo{year}{1978}),
  \bibinfo{pages}{290--307}.
\newblock


\bibitem[Moulin and Shenker(2001)]%
        {MS01}
\bibfield{author}{\bibinfo{person}{Herve Moulin} {and} \bibinfo{person}{Scott
  Shenker}.} \bibinfo{year}{2001}\natexlab{}.
\newblock \showarticletitle{Strategyproof sharing of submodular costs: Budget
  balance vs efficiency}.
\newblock \bibinfo{journal}{\emph{Economic Theory}}  \bibinfo{volume}{18}
  (\bibinfo{year}{2001}), \bibinfo{pages}{511--533}.
\newblock


\bibitem[Nakamoto(2008)]%
        {Nakamoto2008}
\bibfield{author}{\bibinfo{person}{Satoshi Nakamoto}.}
  \bibinfo{year}{2008}\natexlab{}.
\newblock \bibinfo{title}{Bitcoin: A peer-to-peer electronic cash system.}
\newblock \bibinfo{howpublished}{http://bitcoin.org/bitcoin.pdf}.
\newblock


\bibitem[Rey and Rothe(2014)]%
        {rey2014false}
\bibfield{author}{\bibinfo{person}{Anja Rey} {and} \bibinfo{person}{Joerg
  Rothe}.} \bibinfo{year}{2014}\natexlab{}.
\newblock \showarticletitle{False-name manipulation in weighted voting games is
  hard for probabilistic polynomial time}.
\newblock \bibinfo{journal}{\emph{Journal of Artificial Intelligence Research}}
   \bibinfo{volume}{50} (\bibinfo{year}{2014}), \bibinfo{pages}{573--601}.
\newblock


\bibitem[Roughgarden(2016)]%
        {roughgarden16}
\bibfield{author}{\bibinfo{person}{Tim Roughgarden}.}
  \bibinfo{year}{2016}\natexlab{}.
\newblock \bibinfo{booktitle}{\emph{Twenty Lectures on Algorithmic Game
  Theory}}.
\newblock \bibinfo{publisher}{Cambridge University Press}.
\newblock


\bibitem[Rozemberczki et~al\mbox{.}(2022)]%
        {RW+22}
\bibfield{author}{\bibinfo{person}{Benedek Rozemberczki},
  \bibinfo{person}{Lauren Watson}, \bibinfo{person}{P{\'{e}}ter Bayer},
  \bibinfo{person}{Hao{-}Tsung Yang}, \bibinfo{person}{Oliver Kiss},
  \bibinfo{person}{Sebastian Nilsson}, {and} \bibinfo{person}{Rik Sarkar}.}
  \bibinfo{year}{2022}\natexlab{}.
\newblock \showarticletitle{The {S}hapley Value in Machine Learning}. In
  \bibinfo{booktitle}{\emph{Proceedings of the Thirty-First International Joint
  Conference on Artificial Intelligence, {IJCAI} 2022}}.
  \bibinfo{pages}{5572--5579}.
\newblock


\bibitem[Shapley(1953)]%
        {Shapley53}
\bibfield{author}{\bibinfo{person}{Lloyd Shapley}.}
  \bibinfo{year}{1953}\natexlab{}.
\newblock \showarticletitle{A value for $n$-person games}.
\newblock In \bibinfo{booktitle}{\emph{Contributions to the Theory of Games}},
  \bibfield{editor}{\bibinfo{person}{H.~Kuhn} {and} \bibinfo{person}{A.~W.
  Tucker}} (Eds.). Vol.~\bibinfo{volume}{2}. \bibinfo{publisher}{Princeton
  University Press}, \bibinfo{pages}{307--317}.
\newblock


\end{thebibliography}
